\documentclass{lmcs} 
\pdfoutput=1

\usepackage{lastpage}
\lmcsdoi{18}{2}{2}
\lmcsheading{}{\pageref{LastPage}}{}{}%
{Mar.~09,~2021}{Apr.~14,~2022}{}

\usepackage{color}
\usepackage{todonotes}
\usepackage{paralist}
\usepackage{cancel}
\usepackage{xcolor}
\usepackage{amssymb}
\usepackage{amsfonts}
\usepackage{graphicx}
\usepackage{xspace}
\usepackage{tabularx,colortbl}
\usepackage[utf8]{inputenc}


\usepackage{forest}

\forestset{
    treenode/.style = {base=top},
    payoff/.style   = {align=center, base=top}
    }
    
    \usepackage{tikz}

\usepackage{hyperref}
\theoremstyle{plain} 

\newtheorem{definition}{Definition}[section]
\newtheorem{example}{Example}[section]

\newtheorem{remark}{Remark}[section]

\newtheoremstyle{my_theorem}%
  {}
  {}
  {\it}
  {}
  {}
  {}
  { }
  {{\bf \thmname{#1}\thmnumber{ #2}\thmnote{ (#3)}.}}
\theoremstyle{my_theorem}
\newtheorem{theorem}{Theorem}[section]

\newtheorem{lemma}{Lemma}[section]
\newtheorem{proposition}{Proposition}[section]

\newcommand{\EUF}{\ensuremath{\mathcal{EUF}}\xspace}

\newcommand{\dpll}[1]{{\sc DPLL}\xspace}

\newcommand{\COMMENT}[1]{}

\newcommand{\ua}{\ensuremath{\underline a}}
\newcommand{\ub}{\ensuremath{\underline b}}

\newcommand{\ue}{\ensuremath{\underline e}}

\newcommand{\ut}{\ensuremath{\underline t}}
\newcommand{\ux}{\ensuremath{\underline x}}

\newcommand{\uy}{\ensuremath{\underline y}}
\newcommand{\uz}{\ensuremath{\underline z}}

\newcommand{\uw}{\ensuremath{\underline w}}
\newcommand{\uu}{\ensuremath{\underline u}}

\newcommand{\cM}{\ensuremath \mathcal M}
\newcommand{\cN}{\ensuremath \mathcal N}

\newcommand{\cI}{\ensuremath \mathcal I}

\renewcommand{\int}{\ensuremath {\mathcal I}}

\newcommand{\tup}[1]{\langle #1\rangle}            

\begin{document}

\title[Uniform Interpolation Algorithms in \EUF using DAGs]{Uniform Interpolants in \EUF: Algorithms using DAG-representations}

\author[S.~Ghilardi]{Silvio Ghilardi}	
\address{Universit\`a degli Studi di Milano (Italy)}	
\email{silvio.ghilardi@unimi.it}  

\author[A.~Gianola]{Alessandro Gianola}	
\address{Free University of Bozen-Bolzano (Italy)}	
\email{gianola@inf.unibz.it}  

\author[D.~Kapur]{Deepak Kapur}	
\address{University of New Mexico (USA)}	
\email{kapur@cs.unm.edu}  





\begin{abstract}
The concept of  uniform interpolant for a quantifier-free formula from a given formula with a list of symbols, while well-known
in the logic literature, has been unknown  to the formal methods and automated reasoning community for a long time. This concept
is precisely defined. Two algorithms for computing 
quantifier-free uniform interpolants in the theory of equality over uninterpreted symbols (\EUF) endowed with a list of symbols to be eliminated are proposed. The first algorithm is
non-deterministic and generates a uniform interpolant expressed as a disjunction of
conjunctions of literals, whereas the second algorithm gives a
compact representation of a uniform interpolant as a conjunction of Horn clauses. Both algorithms exploit efficient dedicated DAG representations of terms. Correctness and completeness proofs are supplied, using arguments combining rewrite techniques with model theory. 

 \end{abstract}

\maketitle


\section{Introduction}\label{sec:intro}

The theory of equality over uninterpreted symbols, henceforth denoted by \EUF, is one of the simplest theories that have found numerous applications in computer science, compiler optimization, formal methods and logic. Starting with the works of Shostak \cite{Shostak} and Nelson and Oppen \cite{NO} in the early eighties, some of the first algorithms were proposed in the context of developing approaches for combining decision procedures for quantifier-free theories including freely constructed data structures and 
linear
arithmetic over the rationals. \EUF was exploited for hardware verification of pipelined processors by Dill \cite{dill} and more widely subsequently in formal methods and verification using  model checking frameworks. With the popularity of SMT solvers, where \EUF
serves as a glue for combining solvers 
for 
different theories, 
numerous new graph-based algorithms have been proposed in the literature over the last two decades for checking unsatisfiability of a conjunction of (dis)equalities of terms built using function symbols and constants. 

In \cite{McM}, the use of interpolants for automatic invariant generation was proposed, leading to a plethora 
of research activities to develop algorithms for generating interpolants for specific theories as well as their combination. This new application is different from the role of interpolants for analyzing proof theories of various logics starting with the pioneering work of \cite{Craig,HuangInterpolant,Pudlak97} (for a recent survey in the SMT area, see~\cite{bonacina1,bonacina2}). 
Approaches like~\cite{McM,HuangInterpolant,Pudlak97}, 
however, assume access to a proof of 
$\alpha\to \beta$
for which an interpolant is being generated. Given that there can in general be many interpolants including 
infinitely
many for some theories, little is known about what kind of interpolants are effective for different applications, even though some research has been reported on the strength and quality of interpolants \cite{DSilvaKPW10,Weissenbacher12,HoderKV12}. 

In this paper, a different approach is taken, motivated by the insight connecting interpolating theories 
 with those admitting quantifier-elimination, as advocated in \cite{Zarba}. Particularly, in the preliminaries,  
the concept of a uniform interpolant (UI) 
defined by a formula $\alpha$,
in the context of formal methods and verification, is proposed for \EUF, which is well-known not to admit quantifier elimination. 
We recall here uniform interpolants in general; we fix a logic or a theory $T$ and a suitable fragment (propositional, first-order quantifier-free, etc.) of its language $L$.
Given an $L$-formula $\alpha(\ux, \uy)$ (here $\ux,\uy$ are the  variables occurring in $\alpha$), a \emph{uniform interpolant} of $\alpha$ (w.r.t.~$\uy$) is an $L$-formula $\alpha'(\ux)$ where only the $\ux$ occur, and that satisfies the following two properties:
\begin{inparaenum}[\it (i)]
	\item $\alpha(\ux, \uy)\vdash_T \alpha'(\ux)$; 
	\item for any further $L$-formula $\beta(\ux, \uz)$ such that $\alpha(\ux, \uy) \vdash_T \beta(\ux, \uz)$, we have $\alpha'(\ux) \vdash_T \beta(\ux, \uz)$. 
 \end{inparaenum}
Whenever uniform interpolants exist, one can compute an interpolant for an entailment like $\alpha(\ux, \uy) \vdash_T \beta(\ux, \uz)$ in a way that is \emph{independent} of $\beta$. 
A (quantifier-free) uniform interpolant for a formula $\alpha$ is in particular, for any formula $\beta$, an \emph{ordinary} interpolant \cite{Craig,lyndon} for the pair $(\alpha, \beta)$ such that $\alpha \to \beta$
(as well as a \emph{reverse} interpolant \cite{McM} for an unsatisfiable pair $(\alpha, \neg\beta)$).\footnote{The third author recently learned from the first author that this concept has been used extensively in logic for decades \cite{GZ,pitts} to his surprise since he had the erroneous 
impression that he came up with the concept in 2012, which he presented in a series of talks \cite{kapur,kapurJSSC}.} A uniform interpolant could be 
 defined for theories irrespective of whether they admit quantifier elimination. For theories admitting quantifier elimination, a uniform interpolant can be obtained using quantifier elimination: 
  indeed, this shows that a theory enjoying quantifier elimination admits uniform interpolants as well. A uniform interpolant, when it exists, is unique up to logical equivalence: this immediately follows from the definition, since any uniform interpolant implies all the other formulae that are implied by $\alpha$ (and, then implies any other uniform interpolant). 
The equivalent concept of cover is proposed in \cite{GM} (see also \cite{cade19,IJCAR20,JAR21}).

Two algorithms with different characteristics
 for generating uniform interpolants from a formula in  \EUF (with a list of symbols to be eliminated) are proposed in this paper. 
They share a common subpart based on concepts used in a ground congruence closure  proposed in \cite{KapurRTA}, which flattens the input and generates a canonical rewrite system on constants along with  unique rules of the  form $f(\cdots)$, where $f$ is an uninterpreted symbol and the arguments $(\cdots)$ are canonical forms of constants. 
 Further, eliminated symbols are represented as a DAG 
 (`Directed Acyclic Graph')
 to avoid any exponential blow-up. 
 DAG representation of terms are commonly used in theorem provers and are essential to keep basic algorithms like unification polynomial (see e.g. the detailed analysis in textbooks like~\cite{BaNi98}). In this paper, 
 we also introduce a new DAG-representation for terms, called `conditional DAG'-representation, see below. Both DAG- and conditional DAG-representations require an exponential blow-up in order to get full string representations of terms  (precise instructions for the related unravellings will be recalled/supplied in the paper).  However  a careful choice of adequate data structures inside the desired applications can avoid the need of such an unravelling and this is the reason for relying on DAG-compressed representations while designing symbolic algorithms.

Our first algorithm is non-deterministic where undecided equalities on constants are hypothesized to be true or false, generating a branch in each case, and recursively applying the algorithm. 
It could also be formulated as an algorithm similar in spirit to the use of equality interpolants in the Nelson and Oppen framework for combination, where different partitions on constants are tried, with each leading to a branch in the algorithm. New symbols are introduced along each branch 
to avoid exponential blow-up. 
Our second algorithm generalizes the concept of a DAG to conditional DAG in which subterms are replaced by new symbols under a conjunction of equality 
atoms, resulting in its compact and efficient representation. A fully or partially expanded form of a UI can be derived based on their use in applications. Because of their compact representation, 
UIs can be 
kept of polynomial size
for a large class of formulas. 

The former algorithm is tableaux-based, in the sense that it has a tree structure computed for the logical formulae involved and a finite collection of rules that specifies how to create  branches. This algorithm produces the output in disjunctive normal form. The second algorithm is based on manipulation of Horn clauses and gives the output in (compressed) conjunctive normal form. We believe that the two algorithms are 
complementary to each other, especially from the point of view of applications. Model checkers typically synthesize safety invariants using conjunctions of clauses and in this sense they might better take profit from the second algorithm; however, model-checkers
dually representing sets of backward reachable states as disjunctions of \emph{cubes} (i.e., conjunctions of literals) 
would better adopt the first algorithm. Non-deterministic manipulations of cubes are also required to match certain PSPACE lower bounds, as in the case of SAS systems mentioned in~\cite{MSCS20}. On the other hand, regarding the overall complexity, it seems to be easier to avoid exponential blow-ups in concrete examples by adopting the second algorithm.

The termination, correctness and completeness of both the algorithms are proved using results in model theory about model completions; this relies on a basic result (Lemma~\ref{lem:cover} below) taken from \cite{cade19}. 

Both our algorithms are simple,
intuitive and easy to understand in contrast to other 
algorithms in the literature. In fact, the
algorithm from~\cite{cade19} requires the full saturation of all the formulae deductively implied 
in a version of superposition calculus equipped with ad hoc settings 
(in that context, no compact representation of the involved formulae is considered), 
whereas the main merit of our second algorithm is to show that a very light form of completion is sufficient, thus simplifying the whole procedure and getting seemingly better complexity results.\footnote{Although we feel that some improvement is possible, the termination argument in~\cite{cade19} gives a double exponential bound, whereas we have a simple exponential bound for both algorithms (with optimal chances to keep the output polynomial in many concrete cases in the second algorithm).} The algorithm from~\cite{GM} presents some issues/bugs that need to be fixed (see \cite{cade19,JAR21} for details) and the correctness proof has never been published  (the technical report mentioned in \cite{GM} is not available).



The paper is structured as follows: in the next paragraph we discuss about related work on the use UIs. In Section~\ref{sec:prelim} we state the main problem,
settle on some notation, discuss DAG representations and congruence closure. In Sections~\ref{sec:tab} and~\ref{sec:cond}, we respectively give the two algorithms for computing uniform interpolants in \EUF (correctness and completeness of such algorithms are proved in Section~\ref{sec:completeness}). We conclude in Section~\ref{sec:concl}. This paper extends a conference paper (\cite{CILC20}) in two respects: first, it improves the presentation and includes the full proofs, adding also further explanations; second, it contains additional material including detailed examples and some complexity considerations.

\vspace{3mm}

\noindent\textbf{Related work on the use of UIs.}
The use of uniform interpolants in model-checking safety problems for infinite state systems was already mentioned in~\cite{GM} and further exploited in a recent research line on the verification of data-aware processes \cite{CGGMR19,BPM19,MSCS20,BPM20,BPM21}.  Model checkers need to explore the space of all reachable states of a system; a precise exploration (either forward starting from a description of the initial states or backward starting from a description of unsafe states) requires quantifier elimination. The latter is not always available or might have prohibitive complexity; in addition, it is usually preferable to make over-approximations of reachable states both to avoid divergence and to speed up convergence. One well-established technique for computing over-approximations consists in extracting interpolants from spurious traces, see e.g.~\cite{McM}. 
 One possible advantage of uniform interpolants over ordinary interpolants is that they do not introduce over-approximations and so   abstraction/refinements cycles are not needed in case they are employed (the precise reason for that goes through the connection between uniform interpolants, model completeness and existentially closed structures, see~\cite{MSCS20} for a full account). In this sense, computing uniform interpolants has the same advantages and disadvantages as computing quantifier eliminations, with two remarkable differences. The first difference is that uniform interpolants may be available also in theories not admitting quantifier elimination (\EUF being the typical example); the second difference is that computing uniform interpolants may be tractable when the language is suitably restricted e.g. to unary function symbols
(this was already mentioned in~\cite{GM}, see also Remark~\ref{rmk:arity1} below). Restriction to unary function symbols is sufficient in database driven verification to encode primary and foreign keys~\cite{MSCS20}. It is also worth noticing that, precisely by using uniform interpolants for this restricted language, in~\cite{MSCS20} \emph{new decidability results} have been achieved for interesting classes of infinite state systems. 
Notably, such results are also operationally mirrored in the \textsc{MCMT} \cite{mcmt} implementation since version 2.8.


\section{Preliminaries}
\label{sec:prelim}

We adopt the usual first-order syntactic notions, including signature, term,
atom, (ground) formula; 
 our signatures are always \emph{finite} or \emph{countable}
and include equality. 
Without loss of generality, 
only functional signatures, i.e. signatures whose only predicate symbol is equality, are considered. 
A tuple $\tup{x_1,\ldots,x_n}$ of variables is compactly represented as $\ux$. The notation $t(\ux), \phi(\ux)$ means that the term $t$, the formula $\phi$ has free variables included in the tuple $\ux$.
This tuple is assumed to be formed by \emph{distinct variables}, thus we underline that, when we write e.g. $\phi(\ux, \uy)$, we mean that the tuples $\ux, \uy$ are made of distinct variables 
that are also disjoint from each other.
A formula is said to be \emph{universal} (resp., \emph{existential}) if it has the form $\forall \ux (\phi(\ux))$ (resp., $\exists \ux (\phi(\ux))$), where $\phi$ is  quantifier-free. Formulae with no free variables are called \emph{sentences}. 

From the semantic side, 
the standard notion of $\Sigma$-structure $\cM$ is used: this is  a pair formed of a set (the `support set', indicated as $\vert \cM\vert$) and of an interpretation function. The interpretation function maps $n$-ary function symbols to $n$-ary operations on $\vert \cM\vert$  (in particular, constants symbols are mapped to elements of $\vert \cM\vert$).
A free variables assignment $\cI$ on  $\cM$ extends the interpretation function by mapping also variables to elements of $\vert \cM\vert$; the 
notion  of truth of a formula in a $\Sigma$-structure under a free variables assignment $\cI$ is the standard one. 

It may 
be necessary to expand a signature $\Sigma$ with a fresh name for every $a\in \vert \cM\vert$: such expanded signature is called $\Sigma^{|\cM|}$ and $\cM$ is by abuse seen as a $\Sigma^{|\cM|}$-structure itself by interpreting the name of $a\in \vert \cM\vert$ as $a$ (the name of $a$ is
directly indicated as $a$ for simplicity).

A \emph{$\Sigma$-theory} $T$ is a set of $\Sigma$-sentences; a \emph{model}  of $T$ is a $\Sigma$-structure $\cM$ where all sentences in $T$ are true.
	 We use the standard notation $T\models \phi$ to say that $\phi$ is true in all models of $T$ for every assignment to the variables occurring free in $\phi$. We say that $\phi$ is \emph{$T$-satisfiable} iff there is a model $\cM$ of $T$ and an assignment to the variables occurring free in $\phi$ making $\phi$ true in $\cM$.

	 A {\it $\Sigma$-embedding}~\cite{CK} (or, simply, an embedding) between two $\Sigma$-structu\-res $\cM$ and
$\cN$ is a map $\mu: \vert \cM \vert \longrightarrow \vert \cN\vert $ among the
support sets $\vert \cM \vert $ of $\cM$ and $\vert \cN \vert$  of $\cN$ satisfying the condition
$(\cM \models \varphi \quad \Rightarrow \quad \cN \models \varphi)$
for all $\Sigma^{\vert \cM\vert}$-literals $\varphi$ ($\cM$ is regarded as a
$\Sigma^{\vert \cM\vert}$-structure, by interpreting each additional constant $a\in
\vert \cM\vert $ into itself and $\cN$ is regarded as a $\Sigma^{\vert \cM\vert}$-structure by
interpreting each additional constant $a\in \vert \cM\vert $ into $\mu(a)$). 
If  $\mu: \cM \longrightarrow \cN$ is an embedding that is just the
identity inclusion $\vert \cM\vert\subseteq\vert \cN\vert$, we say that $\cM$ is a {\it
	substructure} of $\cN$ or that $\cN$ is an {\it extension} of
$\cM$.

	 Let $\cM$ be a $\Sigma$-structure. The \textit{diagram} of $\cM$, written $\Delta_{\Sigma}(\cM)$ (or just $\Delta(\cM)$), is the set of ground $\Sigma^{|\cM|}$-literals 
that are true in $\cM$. 
%
An easy but important result, called 
\emph{Robinson Diagram Lemma}~\cite{CK},
says that, given any $\Sigma$-structure $\cN$,  the embeddings $\mu: \cM \longrightarrow \cN$ are in bijective correspondence with
expansions of $\cN$ 
to   $\Sigma^{\vert \cM\vert}$-structures which are models of 
$\Delta_{\Sigma}(\cM)$. The expansions and the embeddings are related in the obvious way: 
the name of $a$
is interpreted as $\mu(a)$. 
The typical use of the Robinson Diagram Lemma is the following: suppose we want to show that some structure $\cM$ can be embedded into a structure $\cN$ in such a way that some set of sentences $\Theta$ are true.  Then, by the Lemma, this turns out to be equivalent to the fact that the set of sentences $\Delta(\cM)\cup \Theta$ is consistent: thus, the Diagram Lemma can be used to transform an \emph{embeddability} problem into a \emph{consistency} problem (the latter is a problem of a logical nature, to be solved for instance by appealing to the compactness theorem for first-order logic).

	 %
	%
	%
    %
    
%
%
	
\subsection{ Uniform Interpolants}	
Fix a  theory $T$ and an existential formula $\exists \ue\, \phi(\ue, \uz)$;  
 call a \emph{residue} of $\exists \ue \,\phi(\ue,\uz)$ any quantifier-
free formula $\theta(\uz,\uy)$ such that $T \models \exists
\ue\,\phi(\ue, \uz) \to \theta(\uz, \uy)$ (equivalently, such that
$T \models \phi(\ue, \uz) \to \theta(\uz, \uy)$). The set of residues
of $\exists \ue\, \phi(\ue,\uz)$ is denoted as $Res(\exists \ue\,
\phi(\ue, \uz))$.
%
A quantifier-free formula $\psi(\uz)$  
is said to be a \emph{$T$-(quantifier-free) uniform interpolant}\footnote{In some literature~\cite{GM,cade19} uniform interpolants are called \emph{covers}.} (or, simply, a \emph{uniform interpolant}, abbreviated UI) of $\exists \ue\, \phi(\ue,\uz)$ iff  $\psi(\uz)\in Res(\exists \ue\, \phi(\ue, \uz))$ and $\psi(\uz)$ implies (modulo $T$) all the formulae in $Res(\exists \ue\, \phi(\ue, \uz))$. It is immediately seen that  UIs are unique (modulo $T$-equivalence).
%
A theory $T$ has \emph{uniform quantifier-free interpolation} iff every existential formula $\exists \ue\, \phi(\ue,\uz)$ 
has a UI.

\begin{example}  
Consider the existential formula $\exists e \,(f(e,z_1)=z_2 \land f(e,z_3)=z_4)$: it can be shown that its $\EUF$-uniform interpolant is 
$z_1=z_3\to z_2=z_4$.
\end{example}

Notably, if $T$ has uniform quantifier-free interpolation, then it has ordinary quantifier-free interpolation, 
in the sense that if we have $T\models \phi(\ue, \uz)\to \phi'(\uz, \uy)$ (for quantifier-free formulae $\phi, \phi'$), then there is a quantifier-free formula $\theta(\uz)$ such that  $T\models \phi(\ue, \uz)\to \theta(\uz)$ and $T\models \theta(\uz)\to \phi'(\uz, \uy)$. In fact, if $T$ has uniform quantifier-free interpolation, then the interpolant $\theta$ is independent on $\phi'$ (the  same $\theta(\uz)$ can be used as interpolant for all entailments $T\models \phi(\ue, \uz)\to \phi'(\uz, \uy)$, varying $\phi'$).  
Uniform quantifier-free interpolation has a direct connection to an important notion from classical model theory, namely model completeness (see~\cite{cade19} for more information).

\subsection{Problem Statement}

In this paper 
 the problem of \emph{computing UIs for the case in which $T$ is pure identity theory in a functional signature $\Sigma$} is considered; this theory is called $\EUF(\Sigma)$ or just $\EUF$ in the SMT-LIB2 terminology. 
 Two different algorithms are proposed for that (while proving correctness and completeness of such algorithms, 
 it is simultaneously shown that UIs exist in \EUF). The first 
algorithm computes a UI in disjunctive normal form format, whereas the second algorithm supplies a UI in conjunctive normal form format. Both algorithms use suitable DAG-compressed representation of formulae.

The following notation is used throughout the paper. Since it is easily seen that existential quantifiers 
 commute with disjunctions, it is sufficient to compute UIs for \emph{primitive} formulae, i.e. for formulae of the kind $\exists \ue\, \phi(\ue,\uz)$, where $\phi$ is a \emph{constraint}, i.e. a conjunction of literals. 
We partition all the $0$-ary 
symbols from the input as well as symbols newly introduced into
disjoint sets.
We use the following conventions: 
\begin{compactenum}
 \item[-] $\ue=e_0, \dots, e_N$ (with $N$ integer) are  symbols to be eliminated, called \emph{variables},
 \item[-] $\uz=z_0, \dots, z_M$ (with $M$ integer) are  symbols \emph{not} to be eliminated, called \emph{ parameters},
 \item[-] symbols $a, b, \dots$ stand for both variables and parameters, and for (fresh) constants as well (usually introduced during skolemization).  
\end{compactenum}
In the following we will also use symbols $\uy$ for indicating variables that changed their status and do not need to be eliminated anymore: we use symbols $a, b, \dots$ for them as well. 
Variables $\ue$ are usually skolemized during the manipulations of our algorithms and proofs below, in the sense that 
they have to be considered as fresh individual constants. 

\begin{remark} UI computations eliminate symbols which are existentially quantified variables (or skolemized constants).  \emph{Elimination of function symbols} can be reduced to elimination of variables in the following way. Consider a formula $\exists f\, \phi(f,\uz)$, where $\phi$ is quantifier-free. Successively abstracting out functional terms, we get that  $\exists f\, \phi(f,\uz)$ is equivalent to a formula of the kind
$\exists \ue \,\exists f (\bigwedge_i (f(\ut_i)=e_i) \wedge \psi)$, where the $\ue$ are fresh variables  (with $e_i\in \ue$), $\underline{t}_i$ are terms, 
$f$ does not occur in  
 $\ut_i, e_i, \psi$  and $\psi$ is quantifier-free. The latter is semantically equivalent to  
 $\exists \ue (\bigwedge_{i\neq j} (\ut_i = \ut_j \to e_i= e_j) \wedge \psi)$, where  $\ut_i = \ut_j$ 
 is
 the conjunction of the component-wise equalities of the tuples $\ut_i$ and $\ut_j$.
\end{remark}

\subsection{Flat Literals, DAGs and Congruence Closure}\label{sec:CC}

A \emph{flat literal} is a literal of one of the following kinds
\begin{equation}\label{eq:lit}
f(a_1,\dots, a_n)= b, \quad a_1=a_2, \quad a_1\neq a_2
\end{equation}
where $a_1, \dots, a_n$ and $b$ 
 are (not necessarily distinct) variables or constants. A formula is flat iff all literals occurring in it are flat; flat terms are 
terms that may occur in a flat literal (i.e. terms like those appearing in~\eqref{eq:lit}). 

We call a \emph{DAG-definition} (or simply a DAG) any formula $\delta(\uy, \uz)$ of the following form 
 (where $\uy:=y_1 \dots, y_n$):
$$
\bigwedge_{i=1}^n (y_i=f_i(y_1, \dots, y_{i-1}, \uz))~.    
$$
Thus, $\delta(\uy, \uz)$ provides 
%
an  \emph{explicit definition} of the $\uy$ in terms of the parameters $\uz$.  

Given a DAG $\delta$, we can in fact associate to it 
the substitution $\sigma_{\delta}$ recursively defined by the mapping 
$$(y_i)\sigma_{\delta}~:=~ f_i((y_1)\sigma_{\delta}, \dots, (y_{i-1})\sigma_{\delta}, \uz).  
$$  
%
%
DAGs are commonly used to represent formulae and substitutions in compressed form: in fact a formula like
\begin{equation}\label{eq:constraint}
\exists \uy ~(\delta(\uy, \uz) \wedge \Phi(\uy, \uz))   
\end{equation}
is equivalent to $\Phi((\uy)\sigma_{\delta}, \uz)$, and is called  \emph{ DAG-representation }.  
The formula $\Phi((\uy)\sigma_{\delta}, \uz)$ is said to be the \emph{unravelling} of~\eqref{eq:constraint}: notice that computing  such an unravelling 
in uncompressed form by explicitly performing substitutions
causes an exponential blow-up. This is why we shall syste\-ma\-tically prefer  DAG-representations~\eqref{eq:constraint} 
 to their uncompressed forms.

As above stated, our main aim is to compute the UI of a primitive formula $\exists \ue\, \phi(\ue,\uz)$;  
using trivial logical manipulations (that have just linear complexity costs), it can be shown that, without loss of generality 
\emph{the constraint $\phi(\ue,\uz)$ can be assumed to be flat}. To do so, it is sufficient to perform a \emph{preprocessing procedure} 
by applying well-known Congruence Closure Transformations: the reader is referred to~\cite{KapurRTA} for a full account. 


\section{The Tableaux  Algorithm}\label{sec:tab}

The algorithm proposed in this section is tableaux-like. It manipulates formulae in the following \emph{DAG-primitive} format
\begin{equation}\label{eq:primitive}
\exists \uy ~(\delta(\uy, \uz) \wedge \Phi(\uy, \uz)\wedge \exists \ue ~\Psi(\ue, \uy, \uz))~~
\end{equation}
where $\delta(\uy, \uz)$ is a DAG and $\Phi, \Psi$ are flat constraints (notice that the $\ue$ do not occur in $\Phi$). We call a formula of that format a \emph{DAG-primitive formula}. 
To make reading easier, we shall omit in~\eqref{eq:primitive} the existential quantifiers, so as~\eqref{eq:primitive} will be written simply as
\begin{equation}\label{eq:primitive1}
\delta(\uy, \uz) \wedge \Phi(\uy, \uz)\wedge \Psi(\ue, \uy, \uz)~~.
\end{equation}

We remark that $\Psi$ can contain literals whose terms depend explicitly on $\ue$, whereas $\Phi$ does \emph{not} contain \emph{any} occurrence of the $\ue$ variables.
Initially the DAG $\delta$ and the constraint $\Phi$ are the empty conjunction.
In the DAG-primitive formula~\eqref{eq:primitive1}, 
variables $\uz$ are called \emph{parameter} variables, variables $\uy$ are called \emph{ (explicitly) defined} variables
and variables $\ue$ are called \emph{(truly) quantified} variables. Variables $\uz$ are never modified; in contrast, during the execution of the algorithm it could happen that some quantified variables
may disappear or become defined variables (in the latter case they are renamed: a quantified variables $e_i$ becoming defined is renamed as 
$y_j$, for a fresh $y_j$). Below,  letters $a,b, \dots$ range over $\ue\cup \uy\cup \uz$.

\begin{definition}
A term $t$ (resp. a literal $L$) is $\ue$-free when there is no occurrence of any of the variables $\ue$ in $t$ (resp. in $L$).  Two flat terms $t,u$ of the kinds
\begin{equation}\label{eq:comp}
t:= f(a_1, \dots, a_n)\qquad u:=f (b_1, \dots, b_n) 
\end{equation}
are said to be \emph{compatible} iff for every $i=1, \dots, n$, 
either $a_i$ is identical to $b_i$ or
both $a_i$ and $b_i$ are $\ue$-free.
The \emph{difference set} of two compatible terms as above is the set of disequalities $a_i \neq b_i$, where $a_i$ is not equal to  $b_i$. 
\end{definition}



\subsection{The Tableaux Algorithm}\label{subsec:tabalg}

Our algorithm 
applies the transformations below 
in a ``don't care'' non-deterministic way.
By saying this, we mean that the output is independent (up to logical equivalence) of the order of the application of the transformations, once the following priority  is respected:
the last transformation has lower priority with respect to the remaining transformations. The last transformation is also responsible for  splitting the execution of the algorithm in several branches: 
each branch will produce a different disjunct in the output formula.
Each state of the algorithm is a DAG-primitive formula like~\eqref{eq:primitive1}.
%
%
%
%
We now provide the rules that constitute our `tableaux-like'  algorithm.
\begin{description} \item[{\rm (1)}] \framebox{\emph{ Simplification Rules}:}   
		\begin{description}
			\item[{\rm (1.0)}]  if an atom like $t=t$ belongs to  $\Psi$, just remove it; if a literal like $t\neq t$ occurs somewhere,     delete $\Psi$, replace $\Phi$ with $\bot$ and stop;
			\item[{\rm (1.i)}] If $t$ is not a variable and $\Psi$ contains both $t=a$ and $t=b$, remove the former and replace it with $a=b$. 
			\item[{\rm (1.ii)}] If $\Psi$ contains $e_i=e_j$ with $i>j$, remove it and replace everywhere $e_i$ by $e_j$.
		\end{description}
			\item[{\rm (2)}] \framebox{\emph{DAG Update Rule}:}
			if $\Psi$ contains $e_i=t(\uy, \uz)$, remove it, rename everywhere $e_i$ as $y_j$ (for fresh $y_j$) and add 
			$y_j=t(\uy,\uz)$ to $\delta(\uy, \uz)$. More formally:
			$$  \delta(\uy, \uz) \wedge \Phi(\uy, \uz)\wedge \left( \Psi(\ue, e_i, \uy, \uz)\wedge e_i=t(\uy, \uz)\right) $$
			 $$ \Downarrow $$
			$$  \left(\delta(\uy, \uz) \wedge y_j=t(\uy, \uz)\right) \wedge \Phi(\uy, \uz)\wedge \Psi(\ue, y_j, \uy, \uz)  $$
			\item[{\rm (3)}] \framebox{\emph{$\ue$-Free Literal Rule}:} if $\Psi$ contains a literal $L(\uy, \uz)$, move it to $\Phi(\uy, \uz)$. More formally:
		        $$  \delta(\uy, \uz) \wedge \Phi(\uy, \uz)\wedge \left(\Psi(\ue, \uy, \uz)\wedge L(\uy, \uz)\right) $$
			 $$ \Downarrow $$
			$$  \delta(\uy, \uz) \wedge \left(\Phi(\uy, \uz)\wedge L(\uy, \uz)\right) \wedge \Psi(\ue, \uy, \uz)  $$
			\item[{\rm (4)}] \framebox{\emph{Splitting Rule}:} If $\Psi$ contains a pair of atoms $t=a$ and $u=b$, where $t$ and $u$ are compatible flat terms like in~\eqref{eq:comp},  and no disequality from the difference set of $t,u$ belongs to $\Phi$,
			 then non-deterministically apply one of the following alternatives:
			\begin{description}
			\item[{\rm (4.0)}] remove from $\Psi$
			the atom 
			$f(b_1,\dots, b_n)=b$,  add to $\Psi$ the atom $a=b$ 
			and add to $\Phi$ all equalities $a_i=b_i$ such that $a_i\neq b_i$ is in the difference set of $t,u$;
			\item[{\rm (4.1)}] add to $\Phi$ one of the disequalities from the difference set of $t,u$ (notice that the difference set cannot be empty, otherwise Rule (1.i) applies).
			\end{description}
\end{description}

\noindent When no more rule is applicable, delete $\Psi(\ue, \uy, \uz)$ from the resulting formula $$  \delta(\uy, \uz) \wedge \Phi(\uy, \uz) \wedge \Psi(\ue, \uy,\uz)  $$
so as to obtain for any branch an output formula in DAG-representation of the kind $$  \exists  \uy  ~(\delta(\uy, \uz) \wedge \Phi(\uy, \uz) ) ~~ .$$

We will see in Remark~\ref{rmk:arity1} that in the case of only unary functions, Rule (4) can be disregarded, and the algorithm becomes much simpler and computationally tractable. 

\begin{example} 
We give a simple example for the application of Splitting Rule. Let the pair of atoms $f(z_1,e)=z_2$ and $f(z_3,e)=z_4$ be in $\Psi$ such that 
$z_1\neq z_3$
is not in $\Phi$. Since the difference set of
$f(z_1,e)$ and $f(z_3,e)$ 
is $\{ z_1\neq z_3 \}$ and  
$z_1\neq z_3
\not\in \Phi$, 
Splitting Rule applies. Applying the first alternative (Rule (4.0)), the atom $f(z_3,e)=z_4$ is removed from $\Psi$, and the atoms $z_1=z_3$ and $z_2=z_4$ are added to $\Psi$. Applying the second alternative (Rule (4.1)), the disequality $z_1\neq z_3$ is added to $\Psi$.
\end{example}

Notice that in the above example the literals added to $\Psi$ as a consequence of the Splitting Rule can  then immediately be moved to $\Phi$ via Rule (3), as they are $\ue$-free; this is always the case for 
the disequalities from difference sets in Rule (4.0) and Rule (4.1)  (because they only involve $\ue$-free terms, by definition of difference set), but not necessarily for the atom $a=b$ mentioned in Rule (4.0), because this atom  might not be $\ue$-free.

\begin{remark}\label{rmk:4-rule} 
Splitting Rule (4)  creates branches in an optimized way: indeed, the alternatives $a_i=b_i$ in (4.0) and  $a_i\neq b_i$ in (4.1) are not generated
\emph{for every pair of variables} $a,b$, but the rule is applied only when the left members are \emph{compatible flat terms}. Notice also that in case for every $a_i$ in 
$t$ and for every $b_i$ in $u$ the pairs $a_i,b_i$ are all identical, Rule (1.i) applies instead. 
\end{remark}

The following remark will be useful to prove the correctness of our algorithm, since it gives a description of the kind of literals contained in a state triple that is \emph{terminal} (i.e., when no rule applies).

\begin{remark}\label{rmk:terminal}
Notice that if no transformation applies to~\eqref{eq:primitive}, the set $\Psi$
can only contain disequalities of the kind $e_i\neq a$,  together with equalities of the kind $f(a_1, \dots, a_n)= a$.
However, when it contains $f(a_1, \dots, a_n)= a$, one of the $a_i$ must belong to $\ue$ (otherwise (2) or (3) applies). Moreover, if
$f(a_1, \dots, a_n)= a$ and $f(b_1, \dots, b_n)= b$ are both in $\Psi$, then either they are not compatible or $a_i\neq b_i$ belongs to $\Phi$ for some $i$  and for some variables $a_i, b_i$ not in $\ue$ (otherwise (4) or (1.i) applies).
\end{remark}


The following proposition 
states that, by applying the previous rules, termination is always guaranteed.

\begin{proposition}\label{prop:term-first}
The non-deterministic procedure presented above always terminates.
\end{proposition}

\begin{proof} It is sufficient to show that every branch of the algorithm must terminate. In order to prove that, first observe that the total number of the variables involved never increases and it decreases if  
 (1.ii) is applied (it might decrease also by the effect of (1.0)). Whenever such a number does not decrease, there is a bound on the number of disequalities that can occur in $\Psi, \Phi$. Now transformation (4.1) decreases the number of disequalities that are 
 actually
 \emph{missing}; the other transformations do not increase this number. Finally, all  transformations except (4.1) reduce the  length of  $\Psi$.  
\end{proof}


\begin{remark}\label{rmk:complexity}
 The overall complexity of the above algorithm is exponential in time, because of the number of branches created by Splitting Rule (4). However, the number of rules applied in a single branch is quadratic in time in the dimension of the input: this fact can be proved by relying on the termination argument shown in Proposition \ref{prop:term-first}. Indeed, every rule of the algorithm above except for Rule (4.1) reduces the  length of  $\Psi$, which has length $O(n)$ (where $n$ is the dimension of the input). Let $c_1$ be a counter that decreases every time a rule (except for Rule (4.1)) is applied: hence, $c_1:=O(n)$. Moreover, whenever Rule (4.1) is applied, the length of $\Psi$ remains the same, but the number of disequalities that are actually missing decreases: this number of missing disequalities is clearly $O(n^2)$. Let $c_2$ be a counter that decreases whenever Rule (4.1) is applied: hence, $c_2:=O(n^2)$. Now, consider a counter $c:=c_1+c_2$: this counter decreases every time a rule of the algorithm above is applied. Since $c=c_1+c_2=O(n)+O(n^2)=O(n+n^2)=O(n^2)$, we conclude that the number of rules applied in a single branch is quadratic, as wanted.
 \end{remark}
 
 \begin{remark}\label{rmk:arity1}
 Notice that if function symbols are all unary, 
 there is no need to apply Rule 4.
 Indeed, if all the function symbols are unary, compatible flat terms $u$ and $t$  can either have  the form $u:=f(e), t:= f(e)$ (for some $e$ in $\ue$) or be $\ue$-free. In the former case Rule 1 applies before, whereas in the latter case Rule 3 applies before (we recall that Rule 4 has the lowest priority). Hence, Rule 4 is never applied. Thanks to the argument shown in Remark~\ref{rmk:complexity}, no branch is created and the number of rules applied is quadratic in time in the dimension of the input.
  Hence, for this restricted case computing UI is a tractable problem (this is consistent with \cite{cade19,JAR21}). The case of unary functions has relevant applications in database driven verification~\cite{MSCS20,CGGMR19,BPM19,BPM20,BPM21}
 (where unary function symbols are used to encode primary and foreign keys). 
\end{remark}

\begin{example}\label{ex:first}
 Let us compute the UI of  the  formula $\exists e_0\, (g(z_4,e_0)=z_0 \land f(z_2,e_0)=g(z_3,e_0)\land h(f(z_1,e_0))
 =z_0)$. 
  Flattening gives the set of literals
 \begin{equation}\label{ex:first1}
  g(z_4,e_0)=z_0\land e_1=f(z_2,e_0) \land e_1= g(z_3,e_0)\land e_2=f(z_1,e_0)\land h(e_2)=z_0
 \end{equation}
 where the newly introduced variables 
 $e_1, e_2$
 need to be eliminated too. Applying (4.0) removes
 $g(z_3,e_0)=e_1$ and introduces the new equalities $z_3=z_4$,  $e_1=z_0$. This causes $e_1$ to be renamed as $y_1$ by (2). Applying again (4.0) removes $f(z_1,e_0)= e_2$ and adds the equalities $z_1=z_2$, $e_2=y_1$; moreover, $e_2$ is renamed as $y_2$. To the literal $h(y_2)=z_0$ we can apply (3). The branch terminates with 
 $y_1=z_0\land y_2=y_1 \land z_1=z_2 \land z_3=z_4 \land h(y_2)=z_0 \land  f(z_2,e_0)= y_1 \land g(z_4, e_0)=z_0$. This produces $ z_1=z_2 \land z_3=z_4 \land h(z_0)=z_0$ as a first disjunct of the uniform interpolant. 
 
 The other branches produce $ z_1= z_2\land z_3\neq z_4 $, $ z_1\neq z_2\land z_3= z_4 $  and $ z_1\neq z_2\land z_3\neq z_4 $ as further disjuncts,
 so that the UI turns out to be equivalent  (by trivial logical manipulations) 
 to $ z_1=z_2 \land z_3=z_4 \to h(z_0)=z_0$.
\end{example}

\begin{example}  
Consider the following example, taken from~\cite{GM}.
 We compute the cover of  the primitive formula $\exists e\, (s_1=f(z_3,e)\land s_2=f(z_4,e)\land t=f(f(z_1,e), f(z_2,e)))$, where $s_1, s_2, t$ are terms in $\uz$. 
 
  Flattening gives the set of literals
 \begin{equation}\label{eq:gulv}
 ~f(z_3,e)=s_1~\land ~f(z_4, e)=s_2~ \land ~f(z_1,e)=e_1~\land ~f(z_2, e)=e_2~ \land ~ f(e_1,e_2)=t~
 \end{equation}
 where the newly introduced variables $e_1, e_2$
 need to be eliminated too. We use lists of integers to represent the nodes of the tree created by the tableaux-like algorithm (Figure~\ref{tree1}). 
 
 Applying (4.0) to the first and the second literals of \eqref{eq:gulv}, the first branch is created: node $[1]$ is generated, where $f(z_4,e)=s_2$ is removed and the new equalities $z_3=z_4$,  $s_1=s_2$ are introduced. Then, applying (4.0) to the first and the third literals of \eqref{eq:gulv}, we get a new branch: node $[1.1]$ is generated, where
 $f(z_1,e)=e_1$ is removed and the new equalities $z_3=z_1$,  $s_1=e_1$ are introduced. As shown in Figure~\ref{tree2}, this causes $e_1$ to be renamed as $y_5$ by (2). Applying again (4.0) to the first and the fourth literals of \eqref{eq:gulv}, new branch is created: in node $[1.1.1]$, $f(z_2,e)= e_2$ is removed and  the equalities $z_3=z_2$, $e_2=s_1$ are added; moreover, $e_2$ is renamed as $y_6$ by using (2). In addition, in node $[1.1.2]$, we obtain that no literal is canceled and the inequality  $z_3\neq z_2$ from the difference set of $f(z_3,e)$ and $f(z_2,e)$ is added. To all the newly introduced literals in $\uy, \uz$ in all the branches, we can apply (3). 
 
 The  branch of $[1.1.1]$ terminates with 
 $f(z_3,e)=s_1 \land  f(y_5,y_6)=t \land  z_3=z_4 \land s_1=s_2 \land z_1=z_3 \land s_1=y_5 \land z_3=z_2 \land s_1= y_6$. This produces Leaf $1$: $f(y_5,y_6)=t \land  z_3=z_4 \land s_1=s_2 \land z_1=z_3 \land s_1=y_5 \land z_3=z_2 \land s_1= y_6$ as a first disjunct of the uniform interpolant. 
 
 The  branch of $[1.1.2]$ terminates with 
 $f(z_3,e)=s_1 \land  f(y_5,e_2)=t \land f(z_2,e)=e_2 \land  z_3=y_4 \land s_1=s_2 \land z_1=z_3 \land s_1=y_5 \land z_3\neq z_2$. This produces Leaf $2$: $  z_3=z_4 \land s_1=s_2 \land z_1=z_3 \land s_1=y_5 \land z_3\neq z_2$ as a second disjunct of the uniform interpolant. 
 
 We also analyze a portion of the branch starting with node $[2]$: as a consequence of $(4.1)$, in node $[2]$  the new inequality $z_3\neq z_4$ is introduced. Then, applying (4.0) to the first and the third literals of \eqref{eq:gulv}, we get a new branch: as shown in Figure~\ref{tree3}, node $[2.1]$ is generated, where $f(z_1,e)=e_1$ is removed and the new equalities $z_3=z_1$,  $s_1=e_1$ are introduced.  The last equality causes $e_1$ to be renamed as $y_5$ by (2). Applying again (4.0) to the second and the fourth literals of \eqref{eq:gulv}, a new branch is created: in node $[2.1.1]$, $f(z_2,e)= e_2$ is removed and  the equalities $z_4=z_2$, $s_2=e_2$ are added; moreover, $e_2$ is renamed as $y_6$ by using again (2). To all the newly introduced literals in $\uy, \uz$ in all the branches, we can apply (3). 
 
 The  branch of $[2.1.1]$ terminates with  $f(z_3,e)=s_1 \land f(z_4, e)=s_2 \land f(y_5,y_6)=t  \land z_3\neq z_4 \land z_3=z_1 \land s_1=y_5 \land z_4=z_2\land s_2=y_6$
 This produces Leaf $6$: $ f(y_5,y_6)=t \land z_3\neq z_4 \land z_3=z_1 \land s_1=y_5 \land z_4=z_2\land s_2=y_6$
 as a first disjunct of the uniform interpolant.

 The algorithm generates 16 branches, each of them produces one disjunct of the output formula, so that the UI turns out to be equivalent to 
 $$z_3=z_4\to s_1=s_2 \land \bigwedge_{i,j\in \{3,4\}}  (z_1=z_i \land z_2=z_j \to t= f(s_{i-2},s_{j-2})  ~ . $$

 Notice that this is consistent with: 
 \begin{compactenum} 
 \item the formula in Leaf $1$, where $z_3=z_4$ implies $s_1=s_2$, and where $z_1=z_3$ and $z_2=z_4$ implies $t=f(s_1,s_2)$;
 \item  the formula in Leaf $2$, where $z_3=z_4$ implies $s_1=s_2$;
 \item  the formula in Leaf $6$, where  $z_1=z_3$ and $z_2=z_4$ implies $t=f(s_1,s_2)$.
 \end{compactenum}
 
 
\begin{footnotesize}
 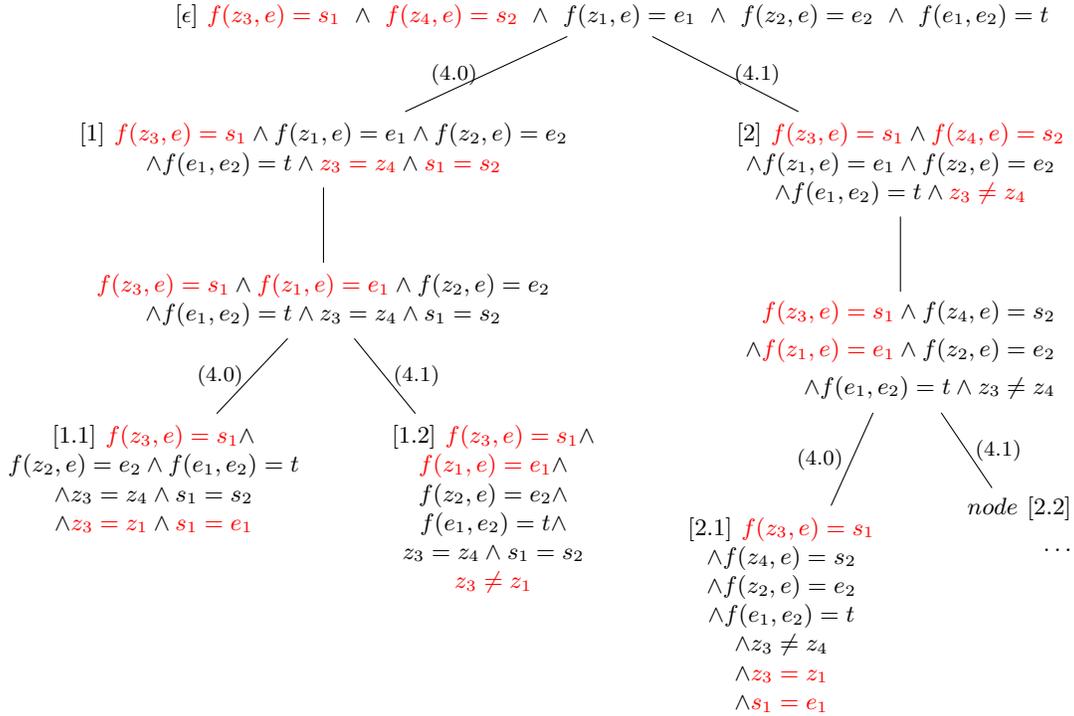
\begin{figure}[htbp]
   \begin{forest}
  for tree={l sep=1cm, s sep=1cm},
    [{$[\epsilon]~\textcolor{red}{f(z_3,e)=s_1}~\land ~\textcolor{red}{f(z_4, e)=s_2}~ \land ~f(z_1,e)=e_1~\land ~f(z_2, e)=e_2~ \land ~ f(e_1,e_2)=t$},treenode
         [{$[1]~\textcolor{red}{f(z_3,e)=s_1}\land f(z_1,e)=e_1 \land f(z_2, e)=e_2$ \\ $\land  f(e_1,e_2)=t \land \textcolor{red}{z_3=z_4} \land \textcolor{red}{s_1=s_2}$}, payoff, edge label={node[midway,left,font=\scriptsize]{$(4.0)$}}
                  [{$\textcolor{red}{f(z_3,e)=s_1}\land \textcolor{red}{f(z_1,e)=e_1} \land f(z_2, e)=e_2$ \\ $\land  f(e_1,e_2)=t \land z_3=z_4 \land s_1=s_2$}, payoff, edge label={node[midway, left, font=\scriptsize]{}}
                  [{$[1.1]~\textcolor{red}{f(z_3,e)=s_1} \land $ \\ $ f(z_2, e)=e_2 \land  f(e_1,e_2)=t $ \\ $ \land z_3=z_4 \land s_1=s_2$ \\ $\land \textcolor{red}{z_3=z_1} \land \textcolor{red}{s_1=e_1}$}, payoff, edge label={node[midway, left, font=\scriptsize]{$(4.0)$}}
                  ]
                  [{$[1.2]~ \textcolor{red}{f(z_3,e)=s_1}\land $ \\ $\textcolor{red}{f(z_1,e)=e_1} \land $ \\ $ f(z_2, e)=e_2 \land $ \\ $  f(e_1,e_2)=t \land  $ \\ $z_3=z_4  \land  s_1=s_2 $ \\ $ \textcolor{red}{z_3 \neq z_1} $}, payoff, edge label={node[midway, right, font=\scriptsize]{$(4.1)$}}]
         ]]
         [{$[2]~\textcolor{red}{f(z_3,e)=s_1}\land \textcolor{red}{f(z_4, e)=s_2}$ \\ $\land f(z_1,e)=e_1 \land f(z_2, e)=e_2$ \\ $\land f(e_1,e_2)=t  \land  \textcolor{red}{z_3\neq z_4}$}, payoff, edge label={node[midway, right, font=\scriptsize]{$(4.1)$}}
         [{$\begin{aligned}~\textcolor{red}{f(z_3,e)=s_1}\land f(z_4, e)=s_2 \\ \land \textcolor{red}{f(z_1,e)=e_1} \land f(z_2, e)=e_2 \\ \land f(e_1,e_2)=t  \land  z_3\neq z_4\end{aligned}$},treenode,edge label={node[midway,right,font=\scriptsize]{}}
          [{$[2.1] ~ \textcolor{red}{f(z_3,e)=s_1}$ \\ $\land f(z_4, e)=s_2$\\  $\land f(z_2, e)=e_2$ \\ $\land f(e_1,e_2)=t$ \\ $\land  z_3\neq z_4$ \\ $\land \textcolor{red}{z_3=z_1}$ \\ $\land \textcolor{red}{s_1=e_1}$}, payoff, edge label={node[midway, left, font=\scriptsize]{$(4.0)$}}]
           [{$\begin{aligned} node\ [2.2]\\ \dots\end{aligned}$},treenode,edge label={node[midway,right,font=\scriptsize]{$(4.1)$}}]
         ]
              ]
    ]
  \end{forest}
 \caption{Portion of the tree created by the tableaux-like algorithm}\label{tree1}
\end{figure}

\end{footnotesize}

 \begin{figure}[htbp]
   \begin{forest}
  for tree={l sep=1.6cm, s sep=2.2cm},
  [{$[1.1]~\textcolor{red}{f(z_3,e)=s_1} \land  f(z_2, e)=e_2 \land  f(e_1,e_2)=t $ \\ $ \land z_3=z_4 \land s_1=s_2\land \textcolor{red}{z_3=z_1} \land \textcolor{red}{s_1=e_1}$}, payoff, edge label={node[midway, left, font=\scriptsize]{$(4.0)$}}
                  [{$\textcolor{red}{f(z_3,e)=s_1} \land \textcolor{red}{f(z_2, e)=e_2} \land  f(y_5,e_2)=t $ \\ $ \land z_3=z_4 \land s_1=s_2 \land z_3=z_1 \land s_1=y_5$}, payoff, edge label={node[midway, left, font=\scriptsize]{$(2)[e_1\to y_5]$}}
                  [{$[1.1.1]~\textcolor{red}{f(z_3,e)=s_1} \land $ \\ $   f(y_5,e_2)=t  \land z_3=z_4 $ \\ $ \land s_1=s_2\land z_3=z_1 $ \\ $ \land  s_1=y_5 \land \textcolor{red}{z_3=z_2} \land \textcolor{red}{s_1=e_2}$}, payoff, edge label={node[midway, left, font=\scriptsize]{$(4.0)$}}
                   [{$[\texttt{leaf}~1]~\xcancel{f(z_3,e)=s_1} \land $ \\ $   f(y_5,y_6)=t  \land z_3=z_4 $ \\ $ \land s_1=s_2\land z_3=z_1 $ \\ $ \land  s_1=y_5 \land z_3=z_2 \land s_1=y_6$}, payoff, edge label={node[midway, right, font=\scriptsize]{$(2)[e_2\to y_6]$}}]
                  ]
                   [{$[1.1.2]~\textcolor{red}{f(z_3,e)=s_1} \land $ \\ $ \textcolor{red}{f(z_2, e)=e_2} \land  f(y_5,e_2)=t $ \\ $ \land z_3=z_4 \land s_1=s_2$ \\ $\land z_3=z_1 \land s_1=y_5 \land \textcolor{red}{z_3\neq z_2}$}, payoff, edge label={node[midway, right, font=\scriptsize]{$(4.1)$}}
                   [{$[\texttt{leaf}~2] ~ \xcancel{f(z_3,e)=s_1} \land $ \\ $ \xcancel{f(z_2, e)=e_2} \land  \xcancel{f(y_5,e_2)=t} $ \\ $ \land z_3=z_4 \land s_1=s_2$ \\ $\land z_3=z_1 \land s_1=y_5 \land z_3\neq z_2 $}, payoff, edge label={node[midway, right, font=\scriptsize]{}}]
                   ]
                   ]
                   ]
 \end{forest}
 \caption{Portion of Branch 1}\label{tree2}
\end{figure}

\begin{footnotesize}
 \begin{figure}[htbp]
   \begin{forest}
  for tree={l sep=1cm, s sep=1.5cm},
                     [{$[2]~\textcolor{red}{f(z_3,e)=s_1}\land \textcolor{red}{f(z_4, e)=s_2}\land f(z_1,e)=e_1 \land f(z_2, e)=e_2$ \\ $\land f(e_1,e_2)=t  \land  \textcolor{red}{z_3\neq z_4}$}, payoff, edge label={node[midway, right, font=\scriptsize]{$(4.1)$}}
         [{$~\textcolor{red}{f(z_3,e)=s_1}\land f(z_4, e)=s_2 \land \textcolor{red}{f(z_1,e)=e_1}$ \\ $\land f(z_2, e)=e_2  \land f(e_1,e_2)=t  \land  z_3\neq z_4$},payoff,edge label={node[midway,right,font=\scriptsize]{}}
          [{$[2.1] ~ \textcolor{red}{f(z_3,e)=s_1}$ \\ $\land f(z_4, e)=s_2 \land f(z_2, e)=e_2$ \\ $\land f(e_1,e_2)=t\land  z_3\neq z_4$ \\ $\land \textcolor{red}{z_3=z_1}\land \textcolor{red}{s_1=e_1}$}, payoff, edge label={node[midway, left, font=\scriptsize]{$(4.0)$}}
          [{$ f(z_3,e)=s_1$ \\ $\land \textcolor{red}{f(z_4, e)=s_2} \land \textcolor{red}{f(z_2, e)=e_2}$ \\ $\land f(y_5,e_2)=t\land  z_3\neq z_4$ \\ $\land z_3=z_1\land s_1=y_5$}, payoff, edge label={node[midway, left, font=\scriptsize]{$(2)[e_1\to y_5]$}}
           [{$[2.1.1]  f(z_3,e)=s_1$ \\ $\land \textcolor{red}{f(z_4, e)=s_2} $ \\ $ \land f(y_5,e_2)=t \land  z_3\neq z_4$ \\ $\land z_3=z_1 \land s_1=y_5$ \\ $ \textcolor{red}{z_4=z_2}\land\textcolor{red}{s_2=e_2} $}, payoff, edge label={node[midway, left, font=\scriptsize]{$(4.0)$}}
           [{$[\texttt{leaf}~6] ~ \xcancel{f(z_3,e)=s_1}$ \\ $\land \xcancel{f(z_4, e)=s_2} $ \\ $ \land f(y_5,y_6)=t \land  z_3\neq z_4$ \\ $\land z_3=z_1 \land s_1=y_5$ \\ $ z_4=z_2\land s_2=y_6  $}, payoff, edge label={node[midway, left, font=\scriptsize]{$(2)[e_2\to y_6]$}}]
	  ]
            [{$[2.1.2] ~ f(z_3,e)=s_1$ \\ $\land \textcolor{red}{f(z_4, e)=s_2} \land \textcolor{red}{f(z_2, e)=e_2}$ \\ $\land f(y_5,e_2)=t\land  z_3\neq z_4$ \\ $\land z_3=z_1\land s_1=y_5 \land \textcolor{red}{z_4\neq z_2 }$}, payoff, edge label={node[midway, right, font=\scriptsize]{$(4.1)$}}
             [{$[\texttt{leaf}~7] ~ \xcancel{f(z_3,e)=s_1}$ \\ $\land \xcancel{f(z_4, e)=s_2} \land \xcancel{f(z_2, e)=e_2}$ \\ $\land \xcancel{f(y_5,e_2)=t}\land  z_3\neq z_4$ \\ $\land z_3=z_1\land s_1=y_5 \land z_4\neq z_2 $}, payoff, edge label={node[midway, right, font=\scriptsize]{$(4.1)$}}]
            ]
            ]
          ]
           [{$[2.2] ~ \textcolor{red}{f(z_3,e)=s_1}\land f(z_4, e)=s_2$ \\ $\land \textcolor{red}{f(z_1,e)=e_1} \land f(z_2, e)=e_2$ \\ $\land f(e_1,e_2)=t\land  z_3\neq z_4\land \textcolor{red}{z_3\neq z_1}$}, payoff, edge label={node[midway, right, font=\scriptsize]{$(4.1)$}} 
           [{$node \ [2.2.1]$\\ $\dots$}, payoff, edge label={node[midway, right, font=\scriptsize]{$(4.0)$}}] 
           [{$node \ [2.2.2]$ \\ $\dots$}, payoff, edge label={node[midway, right, font=\scriptsize]{$(4.1)$}}] 
                 ]
            ]
              ]
 \end{forest}
 \caption{Portion of Branch $2$}\label{tree3}
\end{figure}
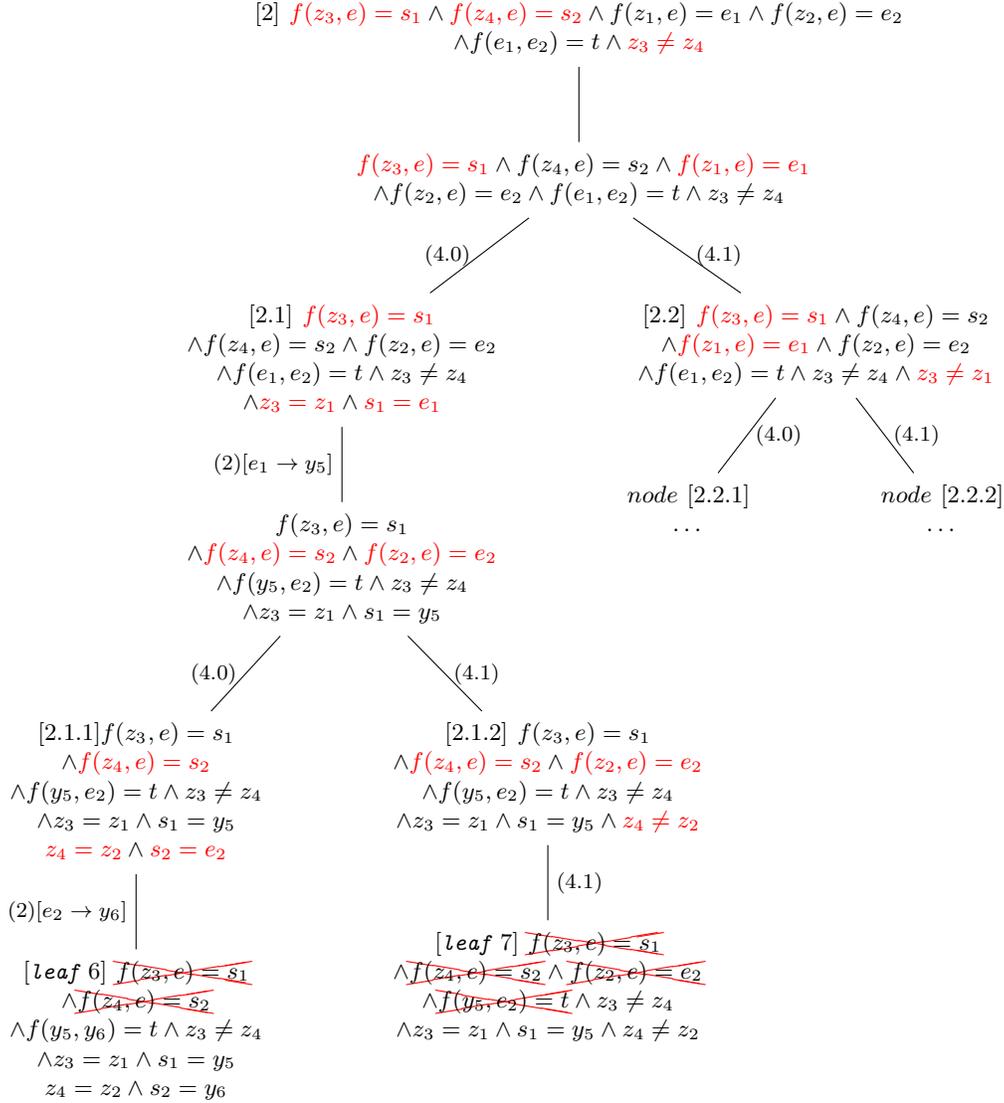
\end{footnotesize}

\end{example}

\section{The Conditional Algorithm}\label{sec:cond}
This section discusses a new algorithm with the objective of
generating a compact representation of the UI in \EUF: this representation avoids
splitting and is based on 
conditions in Horn clauses generated
from literals whose left sides have the same function symbol. 
A
by-product of this approach is that the size of the output UI often can  be kept polynomial. 
%
Further, the output of this algorithm 
generates the UI of $\exists \ue\, \phi(\ue, \uz)$  (where
$\phi(\ue, \uz)$ is a conjunction of literals and $\ue=e_0,\dots,
e_N$, $\uz=z_0, \dots, z_M$, as usual) in conjunctive normal form 
as a conjunction of Horn clauses (we recall that a Horn clause is a disjunction of literals containing \emph{at most one} positive literal).
%
Toward this goal, a new data
structure of a conditional DAG, a generalization of a DAG, is
introduced so as to maximize sharing of sub-formulas.

Using the core preprocessing procedure explained in
Subsection~\ref{sec:CC}, it is assumed that   $\phi$ is the
conjunction $\bigwedge S_1$, where  $S_1$ 
is a set 
of \emph{flat} literals  
containing only literals of the following two  kinds: 
\begin{equation}\label{eq:1}
 f(a_1, \dots, a_h)= a
\end{equation}
\begin{equation}\label{eq:2}
 a\neq b
\end{equation}
(recall that we use letters $a, b, \dots$ for elements of $\ue\cup \uz$). 
Since literals not involving variables to be eliminated or supplying an explicit definition of one of them can be moved directly to the output,
we can assume that variables in $\ue$ must occur in~\eqref{eq:2} and in the left side of~\eqref{eq:1}. 
We do not include equalities like 
$a=e$
because they can be eliminated by replacement.

\subsection{The Conditional Algorithm}\label{subsec:condalg}

The algorithm requires two steps in order to get a set of clauses representing the output in a suitably compressed format.
\vskip 2mm
\framebox{\textbf{Step 1}.} Out of every pair of literals $f(a_1, \dots, a_h)= a$ and 
$f(a'_1, \dots, a'_h)= a'$ of the kind~\eqref{eq:1} (where 
$a$ is syntactically different from $a'$)
we produce the Horn clause
\begin{equation}\label{eq:3}
  a_1=a_1', \dots, a_h=a'_h \to a=a'
\end{equation}
 which can be further simplified by deleting identities in the antecedent. 
Let us call $S_2$ the set of clauses obtained from $S_1$ by adding these new Horn clauses to it.
\vskip 2mm
\framebox{\textbf{Step 2}.}
We saturate $S_2$ with respect to the following rewriting rule

\vskip 2mm
\noindent\begin{tabularx}{\textwidth}{ccl}
\begin{tabular}{@{}c@{}}
\bfseries 
\end{tabular}
&
$\begin{array}{c}
\Gamma \to e_j=e_i\quad\quad C
\\
\hline  \Gamma \to C[e_i]_p 
\\
\end{array}$
&
\\
\end{tabularx}
\vskip 2mm\noindent
where $j>i$; $C[e_i]_p$ means the result of the replacement of $e_j$ by $e_i$ in the position $p$ of the clause $C$ (for some position $p$ in $C$ where $e_j$ appears) 
 and
$\Gamma \to C[e_i]_p$ is the clause obtained by merging $\Gamma$ with the antecedent of the clause 
$C[e_i]_p$.\footnote{ 
Negative literals like~\eqref{eq:2} are  Horn clauses of the kind $a=b\to$  and in Step 2 may be responsible of producing non-unit negative Horn clauses.
}

Notice that we apply the rewriting rule only to conditional equalities of the kind $\Gamma \to e_j=e_i$: this is because clauses like $\Gamma \to e_j=z_i$ are considered `conditional definitions' (and the clauses like $\Gamma \to z_j=z_i$ as `conditional facts').

We let $S_3$ be the set of clauses obtained from $S_2$ by saturating it with respect to the above rewriting rule, by removing from antecedents identical literals of the kind $a=a$ and by removing subsumed clauses. 

\begin{example}\label{ex:cdag}
 Let $S_1$ be the set of the following literals
 \begin{eqnarray*}
  f_1(e_0, z_1)= e_1,~~   
  & 
  f_1(e_0, z_2)= z_3,~~
  &
  f_2(e_0, z_4)= e_2,   
  \\
  f_2(e_0, z_5)= z_6,~~
  &
  g_1(e_0, e_1)= e_2,~~   
  & 
  g_1(e_0, z_1')= z_2',
  \\
  g_2(e_0, e_2)= e_1,~~   
  & 
  g_2(e_0, z''_1)= z_2''~~
  &
  h(e_1,e_2)= z_0
 \end{eqnarray*}
Step 1 produces the following set $S_2$ of Horn clauses
\begin{eqnarray*}
&
 z_1=z_2 \to e_1= z_3,~~
  &
 z_4=z_5 \to e_2= z_6,
 \\ &
 e_1=z'_1 \to e_2= z'_2,~~
  &
 e_2=z''_1 \to e_1= z''_2
\end{eqnarray*}
Since there are no Horn clauses whose consequent is  an equality of the kind $e_i=e_j$, Step 2 does not produce further clauses and we have $S_3=S_2$.
\end{example}



\subsection{Conditional DAGs}

In order to be able to extract the output UI in a uncompressed format 
out of the above set of clauses $S_3$, we must identify all the  `implicit conditional definitions' it contains. As for illustration, Example~\ref{ex:cdag} contains, among the others, the following `implicit' conditional definitions: the variable $e_1$ can be conditionally defined with $e_1:= z_3$, when $z_1=z_2$ holds (because $z_1=z_2\to e_1= z_3$ is in $S_3$); moreover,  when $e_1$ becomes conditionally defined, also $e_2$ becomes implicitly defined as $e_2:= z'_2$, under the condition that $e_1=z'_1$ holds (because $e_1=z'_1 \to e_2= z'_2$ is in $S_3$ as well). All such conditional definitions need to be made explicit, and this is what we are going to formalize in the following.  

Let $\uw$ be an ordered subset of the $\ue=\{e_1, \dots, e_N\}$: that is, in order to specify $\uw$ we must take a subset of the $\ue$ and an ordering of this subset. Intuitively, these $\uw$ will play the role of \emph{placeholders} inside  a conditional definition. 

If we let $\uw$ be $w_1, \dots, w_s$ (where, say,
$w_i$ is some $e_{k_i}$ with $k_i\in \{1, \dots, N\}$), we let $L_i$ be the language restricted to $\uz$ and $w_1, \dots, w_i$ (for $i\leq s$): in other words, an $L_i$-term or an $L_i$-clause may contain only terms built up from $\uz,w_1, \dots, w_i$ by applying to them function symbols. In particular, $L_s$ (also called $L_{\uw}$) is the language restricted to $\uz\cup \uw$. We let $L_0$ be the language restricted to $\uz$.

Given a set $S$ of clauses and $\uw$ as above, a \emph{$\uw$-conditional} DAG $\delta$ (or simply a conditional DAG $\delta$) built out of $S$ is 
an $s$-tuple of Horn clauses from $S$
\begin{equation}\label{eq:cDAG}
\Gamma_1 \to w_1 = t_1, ~\dots,~ \Gamma_s \to w_s = t_s
\end{equation}
where $\Gamma_i$ is a finite tuple of $L_{i-1}$-atoms and $t_i$ is an 
$L_{i-1}$-term. 
Intuitively, a conditional DAG takes into consideration the given ordered subset $\uw$
of the $\ue$ symbols and on top of that it builds a sequence of conditional equations (i.e., Horn clauses), where step by step more and more $\ue$ symbols are employed and iteratively defined in terms of previously already defined $\ue$ symbols that are precedent in the order. Roughly speaking, each set of ``dependencies'' induced among the $\ue$-symbols provides a suitable set of (conditional) definitions. 

Conditional DAGs are used to define suitable formulae that are needed for the construction of the uniform interpolant. We now define such formulae.  
Given a $\uw$-conditional DAG $\delta$ we can define the  formulae 
$\phi^i_{\delta}$ (for $i=1, \dots, s+1$) as follows:
\begin{compactenum}
 \item[-] $\phi_{\delta}^{s+1}$ is the conjunction of all $L_{\uw}$-clauses belonging to $S$;
 \item[-] for $i\leq s$, the formula $\phi_{\delta}^{i}$ is $\Gamma_i \to \forall w_i\, (w_i=t_i \to \phi_{\delta}^{i+1})$. 
\end{compactenum}
It can be seen 
 that $\phi^i_{\delta}$ is equivalent to a quantifier-free $L_{i-1}$ formula,\footnote{
Since $\phi^i_{\delta}$ is logically equivalent to $(\bigwedge \Gamma_i) \to  \phi_{\delta}^{i+1}(t_i/w_i)$, it is immediate to see 
that it can be recursively turned, again up to equivalence, into a conjunction of Horn clauses.
} in particular $\phi^1_{\delta}$ (abbreviated as $\phi_{\delta}$) is equivalent to an $L_0$-quantifier-free formula. The explicit computation of such quantifier-free formulae may however produce an exponential blow-up.    
The intuition behind these constructions is that in order to produce the correct uniform interpolant one needs to consider \emph{all} the possible $L_{\uw}$-clauses in $S$ and then iteratively `eliminate' the  $\ue$ symbols in them by exploiting the conditional definitions of the $\ue$ symbols (determined by the conditional DAG $\delta$) in terms of previously defined $\ue$ symbols in the order of $\uw$.


\begin{example}
Let us analyze the conditional DAG $\delta$ that can be extracted out of 
the set $S_3$ of the  Horn clauses mentioned in Example~\ref{ex:cdag} (we disregard those $\delta$ such that $\phi_{\delta}$ is the empty conjunction $\top$).
%
The $\uw_1$-conditional DAG $\delta_1$ with $\uw_1= e_1, e_2$ 
and conditional definitions 
$$
z_1=z_2 \to e_1= z_3,~~e_1=z'_1 \to e_2= z'_2
$$
where $e_2$ depends upon $e_1$, produces formula $\phi_{\delta_1}$, and similarly   
the $\uw_2$-conditional DAG $\delta_2$ with $\uw_2=e_2, e_1$
and conditional definitions 
$$
z_4=z_5 \to e_2= z_6,~~e_2=z''_1 \to e_1= z''_2.
$$
where $e_1$ depends upon $e_2$, produces formula $\phi_{\delta_2}$: notice that $\phi_{\delta_1}$ and $\phi_{\delta_2}$ are not logically equivalent. 

Indeed, $\phi_{\delta_1}$ is logically equivalent to
\begin{equation}\label{eq:uno} 
z_1=z_2 \wedge z_3=z'_1\to \bigwedge S_3\setminus \{ e_{0}\}[z_3/e_1, z'_2/e_2]~~
\end{equation}
where we used the notation $\bigwedge S_3\setminus \{ e_{0}\}[z_3/e_1, z'_2/e_2]$ to mean the result of the substitution of  $e_1$ with $z_3$ 
  and of  $e_2$  with $z'_2$ 
 in the conjunction of $S_3$-clauses not involving $e_0$.
Notice that, intuitively, this formula is obtained by iteratively defining, step by step, bigger $\ue$ variables in terms of smaller ones: when  $z_1=z_2 \wedge z_3=z'_1$ holds, $e_1$ is conditionally replaced by $z_3$ and $e_2$ by $z'_2$. 

Analogously, $\phi_{\delta_2}$ is logically equivalent to
\begin{equation}\label{eq:due}
z_4=z_5 \wedge z_6=z''_1\to \bigwedge S_3\setminus \{ e_{0}\}[z_6/e_2,z''_2/e_1]
\end{equation}
 (the explanation of the notation  $S_3\setminus \{ e_{0}\}[z_6/e_2,z''_2/e_1]$
 is the same as the explanation for the notation $\bigwedge S_3\setminus \{ e_{0}\}[z_3/e_1, z'_2/e_2]$ used above).
 A third possibility is to use the conditional definitions $z_1=z_2 \to e_1= z_3$ and $z_4=z_5 \to e_2= z_6$ with (equivalently) either $\uw_1$ or $\uw_2$ resulting in a conditional dag $\delta_3$ with  $\phi_{\delta_3}$ logically equivalent to 
 \begin{equation}\label{eq:tre}
 z_1=z_2 \wedge z_4=z_5 \to \bigwedge S_3\setminus \{ e_{0}\}[z_3/e_1,z_6/e_2]~~.
\end{equation}
\end{example}

The next lemma 
shows the relevant property of $\phi_{\delta}$:

\begin{lemma}\label{lem:oneside} For every set of clauses $S$ and for every 
$\uw$-conditional DAG $\delta$ built out of $S$, 
the formula
$
\bigwedge S \to \phi_{\delta}
$
is logically valid.
\end{lemma}

\begin{proof}
 We prove that $ \bigwedge S \to \phi_{\delta}^i$ is valid by induction on $i$. The base case is clear. For the case $i\leq s$, proceed, e.g., in natural deduction as follows: 
 assume $S, \Gamma_i$ and $\tilde w_i=t_i$ in order to prove $\phi_{\delta}^{i+1}(\tilde w_i/w_i)$. Since $\Gamma_i\to w_i=t_i\in S$, then 
 by implication elimination you get $w_i=t_i$ and also $w_i=\tilde w_i$ by transitivity of equality. Now
 you get what you need from induction hypothesis and equality replacement.
\end{proof}

Notice that it is not true that the conjunction of all possible $ \phi_{\delta}$ (varying $\delta$ and $w$) implies $\bigwedge S$: in fact, such a conjunction 
can be empty 
for instance in case $S$ is just $\{e_1=e_2\}$.

\subsection{Extraction of UI's}

We shall prove below that in order to get a UI of $\exists \ue\, \phi(\ue, \ua)$, one can take \emph{the conjunction of all possible $ \phi_{\delta}$, varying $\delta$ among the  conditional DAGs that can be built out of the set of clauses $S_3$}  from Step 2 of the above algorithm. 
We highlight that in order to generate the correct output (i.e., the uniform interpolant), one needs to consider \emph{all} the possible conditional DAGs built out of the set of clauses $S_3$  from Step 2, which implies to take into consideration all DAGs that can be defined considering \emph{all} possible ordered subsets  $\uw$.

\begin{example}
If $\phi$ is the conjunction of the literals of Example~\ref{ex:cdag}, then the conjunction of 
$\eqref{eq:uno}$, $\eqref{eq:due}$ 
and  $\eqref{eq:tre}$ is a UI of $\exists \ue \,\phi$; in fact, no further  non-trivial conditional dag $\delta$ can be  extracted (if we take $\uw=e_1$ or $\uw=e_2$ or $\uw=\emptyset$ to extract $\delta$, then it happens that $\phi_{\delta}$ is the empty conjunction $\top$).
\end{example}

\begin{example} Let us turn to the literals~\eqref{ex:first1} of Example~\ref{ex:first}. Step 1 produces out of them the conditional clauses
\begin{equation}\label{ex:first2}
 z_3=z_4\to e_1=z_0, \qquad z_1=z_2\to e_2=e_1~~.
\end{equation}
Step 2 produces by rewriting the further  clauses $z_1=z_2\to f(z_1, e_0)= e_1$ and $z_1=z_2\to 
h(e_1)=z_0$.
We can extract two conditional DAGs $\delta$  (using both the conditional definitions~\eqref{ex:first2} or just the first one); in both cases $\phi_{\delta}$ is $ z_1=z_2 \land z_3=z_4 \to h(z_0)=z_0$, which is the  UI. 
\end{example}


As it should be evident from the two examples above,
the conditional DAGs representation of the output considerably reduces computational complexity in many cases; this is a clear advantage of the present algorithm over the algorithm from Section~\ref{sec:tab} and over other approaches like, e.g.~\cite{cade19}. 
Still, the next example shows that  in some cases
 the overall complexity remains exponential.

\begin{example}\label{ex:exp}
 Let $\ue$ be $e_0, \dots, e_N$ and let $\uz$ be 
 $
\{z_0, z'_0\}\cup \{z_{i,j}, z'_{i,j} \mid 1\leq i<j\leq N\}.
 $
 Let $\phi(\ue, \uz)$ be the conjunction of the identities $f(e_0,e_1)=z_0$, $f(e_0,e_N)=z'_0$
 and the 
 set of identities
 $
 h_{ij}(e_0,z_{ij})=e_i,~~~  h_{ij}(e_0,z'_{ij})=e_j,
 $
 varying $i, j$ such that $1\leq i<j\leq N$. We now show that applying the conditional algorithm we get an UI which is exponentially long. 

 After applying Step 1 of the algorithm presented in Subsection~\ref{subsec:condalg}, we get the Horn clauses
 $
 z_{ij}=z'_{ij} \to e_i=e_j,
 $
 as well as the clause $e_1=e_N \to z_0=z'_0$. 
 If we now apply Step 2, 
 we can never produce a conditional clause of the kind 
 $\Gamma \to e_i =t$ with $t$ being $\ue$-free (because we can only rewrite some $e_i$ into some $e_j$).  
 Thus no sequence of clauses like~\eqref{eq:cDAG} can be extracted from $S_3$: notice in fact that the term $t_1$ from such a sequence must not contain the variables $\ue$. In other words, the only $\uw$-conditional DAG $\delta$ that can be extracted is based on the empty $\uw\subseteq \ue$ and is empty itself. 
 
 In order to extract the UI, we need to compute the formulae $\phi_{\delta}$ from any $\uw$-conditional DAG $\delta$, which is only one in such a case. However, this unique $\delta$ produces a formula $\phi_{\delta}$ that is 
  quite big: it is the conjunction of  the 
  clauses from $S_3$ where the $\ue$ do not occur 
  ($S_3$ contains in fact $\Gamma\to z_0=z'_0$ for exponentially many $\ue$-free $\Gamma$'s).
  %

  We conclude this example by commenting on the reason why $\phi_{\delta}$ has an exponential size.   
  In fact, for every minimal set of pairs $I\subseteq \{1,\dots, N\}\times \{1, \dots, N\}$ such that the equivalence relation generated by $I$ contains the pair $(1,N)$, we have that $S_3$ contains the clause $\Gamma_I\to  z_0=z'_0$, where $\Gamma_I$ is the set of equalities $z_{ij}=z'_{ij}$ varying $(i,j)\in I$.\footnote{ 
  You can easily find exponentially many such $I$, e.g. by selecting a subset $X$ of $\{1,\dots, N\}$ containing both $1$ and $N$ and letting $I$ be the set 
  $\{(i,j)\mid i<j, i\in X, j\in X\}$.
  }
\end{example}

\section{Correctness and Completeness Proofs}\label{sec:completeness}

In this section we prove correctness and completeness of our two algorithms. To this aim, we need some preliminaries, both from model theory and 
from term rewriting.

Extensions and UI are related to each other by the following result we take from~\cite{cade19}:

\begin{lemma}[Cover-by-Extensions]\label{lem:cover} 
Let $T$ be a first order theory.
A formula $\psi(\uy)$ is a UI in $T$ of $\exists \ue\, \phi(\ue, \uy)$ iff 
it satisfies the following two conditions:
\begin{description}
\item[{\rm (i)}] $T\models  \forall \uy\,( \exists \ue\,\phi(\ue, \uy) \to \psi(\uy))$;
\item[{\rm (ii)}] for every model $\cM$ of $T$, for every tuple of  elements $\ua$ from the support of $\cM$ such that $\cM\models \psi(\ua)$ it is possible to find
  another model $\cN$ of $T$ such that $\cM$ embeds into $\cN$ and $\cN\models \exists \ue \,\phi(\ue, \ua)$.
\end{description}
\end{lemma}

%
%


For term rewriting we refer to a textbook like~\cite{BaNi98}; 
we only recall 
the following  classical result:
\begin{lemma}\label{lem:canonic}
 Let $R$ be a canonical ground rewrite system over a  signature $\Sigma$. Then there is a $\Sigma$-structure $\cM$ such that for every pair of ground terms $t, u$ we have that $\cM\models t=u$ iff the $R$-normal form of $t$ is the same as the $R$-normal form of $u$. Consequently $R$ is consistent with a set of negative literals $S$ iff for every $t\neq u\in S$ the $R$-normal forms of $t$ and $u$ are different.
\end{lemma}


We are now ready to prove  correctness and completeness of our  algorithms.
We first give the relevant intuitions for the proof technique, which is the same for both cases. By Lemma~\ref{lem:cover} characterizing uniform interpolants, what we need  to show is that if a model $\cM$ satisfies the output formula of the algorithm, then it can be extended to a superstructure $\cN$ 
satisfying the input formula of the algorithm. By Robinson Diagram Lemma, this embeddability problem can be transformed into a consistency problem: in order to do so, we show that the Robinson Diagram $\Delta_\Sigma(\cM)$ is consistent with the 
input formula of the algorithm.
%
For our purposes, it is convenient to see $\Delta_{\Sigma}(\cM)$ as a set of flat literals as follows: the positive part of $\Delta_{\Sigma}(\cM)$ contains the $\Sigma^{\vert \cM\vert}$-equalities $f(a_1, \dots, a_n)=b$ which are true in $\cM$ and the negative part of $\Delta_{\Sigma}(\cM)$ contains the $\Sigma^{\vert \cM\vert}$-inequalities $a\neq b$, varying $a,b$ among the pairs of different elements of $\vert\cM\vert $.
The positive part of $\Delta_\Sigma(\cM)$ is a canonical rewriting system (equalities 
like $f(a_1, \dots, a_n)=a$ are obviously oriented from left-to-right) and every term occurring in $\Delta(\cM)$ is in normal form. If an algorithm works properly, 
 it will be possible to see that the completion of the union of  $\Delta_\Sigma(\cM)$ with 
 the input constraint (or with a constraint 
 equivalent to it)
 is trivial and does not produce inconsistencies. To sum up, the completeness proofs of both algorithms require the following technical ingredients: 
 \begin{compactenum}
 \item Lemma~\ref{lem:cover},  for transforming the problem of computing UI into an embeddability problem;
 \item Robinson Diagram Lemma, for turning the previous problem into a consistency problem, which is more tractable;
 \item the completion of the diagram 
 joined with the input constraint,
 so as to get a canonical rewriting system.
 \end{compactenum}

\vspace{2mm}

\noindent\textbf{Correctness and Completeness  of the Tableaux Algorithm}
\vspace{1mm}

In this subsection, we prove the correctness and completeness of the Tableaux Algorithm. 
We first summarize the structure of the proof by commenting on its main steps. As discussed above, the proof of Theorem~\ref{thm:tab-cover} relies on Lemma~\ref{lem:cover}: in order to prove that the output formula is a uniform interpolant, it is sufficient to show that the embeddability conditions stated in Lemma~\ref{lem:cover} hold. This is achieved, thanks to Robinson Diagram Lemma, by showing that  the Robinson Diagram $\Delta_\Sigma(\cM)$, where $\cM$ satisfies the output formula, is consistent with the 
input formula, as manipulated up to logical equivalencies by the algorithm. In the case of the Tableaux Algorithm, after a normalization of 
a rewriting system suitably extending
this Diagram, we show that all the obtained oriented equalities form a canonical rewriting system: this is an immediate consequence of Remark~\ref{rmk:terminal}. We then conclude applying Lemma~\ref{lem:canonic}:
 this lemma allows us to exhibit a model of the canonical rewriting system, showing in turn the consistency of $\Delta_\Sigma(\cM)$ with the input formula.

\begin{theorem}\label{thm:tab-cover} 
Suppose that we apply the  algorithm of Subsection~\ref{subsec:tabalg} to the
 primitive formula 
 $\exists \ue(\phi(\ue, \uz))$ 
 and that the algorithm terminates with its branches 
 in the states 
$$
\delta_1(\uy_1,\uz)\wedge  \Phi_1(\uy_1,\uz) \wedge \Psi_1(\ue_1, \uy_1, \uz),
~~~\dots,~~~ 
\delta_k(\uy_k,\uz)\wedge  \Phi_k(\uy_k,\uz) \wedge  \Psi_k(\ue_k, \uy_k, \uz)
$$
then the UI of $\exists \ue(\phi(\ue, \uz))$ in \EUF is the \emph{unravelling} 
(see Subsection~\ref{sec:CC}) of the formula  
\begin{equation}
 \bigvee_{i=1}^k \exists \uy_i ~(\delta_i(\uy_i, \uz) \wedge \Phi_i(\uy_i, \uz))~~.
\end{equation}
\end{theorem}

\begin{proof} Since $\exists \ue(\phi(\ue, \uz))$ is logically equivalent to $\bigvee_{i=1}^k \exists \uy_i ~(\delta_i(\uy_i, \uz) \wedge \Phi_i(\uy_i, \uz)\wedge \exists \ue_i \Psi_i(\ue_1, \uy_1, \uz))$, it is sufficient to check that if a formula like~\eqref{eq:primitive} is terminal (i.e. no rule applies to it) then its UI is $\exists \uy ~(\delta(\uy, \uz) \wedge \Phi(\uy, \uz))$.
To this aim, we apply Lemma~\ref{lem:cover}: we pick a model  $\cM$ satisfying the formula $\delta(\uy, \uz) \wedge \Phi(\uy, \uz)$ from the output
 via an assignment $\cI$ to the variables $\uy, \uz$\footnote{Actually the values of the  assignment $\cI$ to the $\uz$ uniquely determines the values of  $\cI$ to the $\uy$.} and we show that $\cM$ can be embedded into a model $\cM'$ such that, for a suitable extensions $\cI'$ of $\cI$ to the variables $\ue$, we have that $(\cM', \cI')$ satisfies also $\Psi(\ue, \uy, \uz)$ from the input.
This embeddability problem can be transformed into a consistency problem as follows. In fact,
what we need (by Robinson Diagram Lemma) is to find a model for the following set of literals 
\begin{equation}\label{eq:TR}
\Delta_\Sigma(\cM)\cup \Psi  \cup \{  a= \tilde a\}_{a\in \uy\cup \uz}
\end{equation}
where $\tilde a$ is the value of $a$ under the assignment $\cI$ (here all variables in~\eqref{eq:TR} are seen as constants, so~\eqref{eq:TR} is a set of ground literals).  
We can orient the equalities in~\eqref{eq:TR} by letting function symbols having bigger precedence 
over constants and by letting $a$ having bigger precedence over $\tilde a$. Normalizing~\eqref{eq:TR} replaces $a$ with $\tilde a$ in $\Psi$: call $\tilde \Psi$ the resulting set of literals 
(we conventionally use $\tilde e_i$  as an alias for $e_i$, for all $e_i\in \ue$, to have a uniform notation). 
After this normalization, we show that all oriented equalities in $\Delta_\Sigma(\cM)\cup \tilde\Psi  \cup \{  a= \tilde a\}_{a\in \uy\cup \uz}$ form a canonical rewriting system. This is due to Remark~\ref{rmk:terminal}: in fact  if $f(a_1, \dots, a_n)$ and $f(b_1, \dots, b_n)$ both occur in $\Psi$, it cannot happen that $f(\tilde a_1, \dots, \tilde a_n)$ and $f(\tilde b_1, \dots, \tilde b_n)$
are the same term because the $\ue$ are already in normal form and because $\cM\models \Phi$ (and hence also 
$\cM\models \tilde \Phi$): 
in particular, all disequalities between $\ue$-free constants $a_i\neq b_i$ belonging to $\Phi$ are true in $\cM$.
In addition, if $f(a_1, \dots, a_n)$ occurs in $\Psi$, then one of the $a_i$ belongs to $\ue$, hence rules from $\tilde \Psi$ and $\Delta_\Sigma(\cM)$ cannot superpose.
Since all oriented equalities in $\Delta_\Sigma(\cM)\cup \tilde\Psi \cup \{  a= \tilde a\}_{a\in \uy\cup \uz}$ form a canonical rewriting system, the inequalities in $\Delta_\Sigma(\cM)\cup \tilde\Psi$ are in normal form and  we can apply Lemma~\ref{lem:canonic} to get the desired $\cM'$:  in fact, Lemma~\ref{lem:canonic} provides a $\Sigma$-structure $\cM'$ that  is a model of the canonical rewriting system $\Delta_\Sigma(\cM)\cup \tilde\Psi \cup \{  a= \tilde a\}_{a\in \uy\cup \uz}$, showing in turn the consistency of $\Delta_\Sigma(\cM)$ with the formula $\Psi$ in input. 
\end{proof}

\noindent\textbf{Correctness and Completeness  of the Conditional Algorithm}
\vspace{1mm}
%
 %

In this subsection we provide the full proof of correctness and completeness of the Conditional Algorithm. 
First of all, we briefly present the main ideas behind this proof.
As in the case of the Tableaux algorithm, exploiting Lemma~\ref{lem:cover}, we need to show that the embeddability conditions stated in Lemma~\ref{lem:cover} hold. We do so by using Robinson Diagram Lemma: we prove that  the Robinson Diagram $\Delta(\cM)$, where $\cM$ satisfies the output formula, is consistent with the
input formula. However, in the case of the Conditional Algorithm, the proof of this fact is more involved: indeed, we use the ground Knuth-Bendix completion in order to prove that no inconsistent literal can be produced, and this requires a careful analysis of the equalities that can be generated during the completion. Specifically, a particular attention is needed for equalities involving only symbols from 
a certain subset
$\ue\setminus \uw$ (called $\uu$ in the proof of the theorem): this analysis is carried out in Lemma~\ref{lem:fact} below. The fact that no inconsistency can be produced concludes the proof of the theorem.

In order to prove Lemma~\ref{lem:fact} (used in the proof of Theorem~\ref{thm:main}), we need to show the following preliminary lemma:

\begin{lemma}\label{lem:pers}
If the clauses $\Gamma \to f(a_1, \dots, a_h)= b$ and $\Gamma'\to f(a'_1, \dots, a'_h)= b'$ both belong to the set of clauses $S_3$ obtained after Step 2 in Subsection~\ref{subsec:condalg} 
 %
 and $b$ is not the same term as $b'$, 
 then $S_3$ contains also a clause subsuming the clause
 \begin{equation*}
 \Gamma, \Gamma', a_1=a'_1, \dots, a_h=a'_h \to b=b'
\end{equation*}
\end{lemma}

\begin{proof}
 By induction on the number $K$ of applications of the  rewriting rule of Step 2 needed to derive $\Gamma \to f(a_1, \dots, a_h)= b$ and $\Gamma'\to f(a'_1, \dots, a'_h)= b'$. If $K$ is 0, the claim is clear by the instruction of Step 1. Suppose that $K>0$ and let $\Gamma'\to f(a'_1, \dots, a'_h)= b'$ be obtained from $\Gamma_1\to e_i= e_j$ by rewriting $e_j$ to $e_i$ from some clause $C$. We need to distinguish cases depending on the position $p$ of the rewriting. All cases being treated in the same way, suppose for instance that $p$ is in the antecedent,\footnote{ 
 There is a case, where $p$ is in the consequent,  that is treated in a slightly different way than all the other ones.
 If, using $\Gamma_1\to e_i= e_j$, we rewrite a clause of the form $\Gamma''\to f(a'_1, \dots, a'_n)= e_j$ into
 $\Gamma_1,\Gamma''\to f(a'_1, \dots, a'_n)= e_i$ (so that $b'$ is the same as $e_i$) and if $b$ is $e_j$, 
 then, instead of applying induction, we can directly take $\Gamma_1\to e_i= e_j$ as the clause we are looking for. If $b$ is not $e_j$, induction applies as in all the other cases.
 }  so that $C$ is $\Gamma_2\to f(a'_1, \dots, a'_h)= b'$ and that $\Gamma'\to f(a'_1, \dots, a'_h)= b'$ is $\Gamma_1,\Gamma_2[e_i]_p\to f(a'_1, \dots, a'_h)= b'$. Then by induction hypothesis $S_3$ contains a clause subsuming 
 $$
 \Gamma,  \Gamma_2, a_1=a'_1, \dots, a_h=a'_h \to b=b'
 $$
 and rewriting with $\Gamma_1\to e_i= e_j$ produces 
 $$
 \Gamma, \Gamma_1,  \Gamma_2[e_i]_p, a_1=a'_1, \dots, a_h=a'_h \to b=b'
 $$
 as required.
\end{proof}

We now state and prove the theorem of correctness and completeness of the Conditional Algorithm.

\begin{theorem}\label{thm:main}
 Let $S_3$  be obtained from  $\exists \ue\, \phi(\ue, \uz)$ as in Steps 1-2 of Subsection~\ref{subsec:condalg}. 
 Then the conjunction $\mathcal{C}$ of all possible $ \phi_{\delta}$ (varying $\delta$ among the  conditional DAGs that can be built out of $S_3$) is a UI of $\exists \ue\, \phi(\ue, \uz)$ in \EUF. 
\end{theorem}

\begin{proof}
 We use Lemma~\ref{lem:cover} in order to show that the output $\mathcal{C}$ is the UI of the input formula $\exists \ue\, \phi(\ue, \uz)$. 
 Condition (i) of that Lemma is ensured by Lemma~\ref{lem:oneside} above because $\bigwedge S_3$ is logically equivalent to $\phi$. So let us take a model $\cM$ and elements $\tilde \ua$ from its support such that we have $\cM\models \bigwedge_{\delta} \phi_{\delta}$ under the assignment of the $\tilde \ua$ to the parameters $\uz$. We need to expand it to a superstructure $\cN$ in such a way that we have $\cN \models \bigwedge S_1$, under some assignment to $\uz, \ue$ extending the assignment $\uz\mapsto \tilde \ua$
 (recall that $\bigwedge S_1$ is logically equivalent to $\phi$ too).
 From now on, we consider the assignment $\uz\mapsto \tilde \ua$ fixed, so that when we write $\cM\models C$ for a clause 
 $C(\uz)$ 
 we mean that $\cM\models C$ holds under the assignment $\uz\mapsto \tilde \ua$.
 
 
 Now, we can transform the embeddability problem of finding the aforementioned superstructure $\cN$ into a consistency problem as follows.
 First of all,  notice that every $\uw$-conditional DAG $\delta$ extracted from $S_3$ (let it be given by the clauses~\eqref{eq:cDAG})
 is naturally equipped with a substitution $\sigma_{\delta}$ which is given in DAG form by
$
w_1 \mapsto t_1, ~\dots,~  w_s \mapsto t_s.
$
 We say that $\delta$ is \emph{realized} in $\cM$ iff 
 we have that
 $$
 \cM \models \bigwedge \Gamma_1\sigma_{\delta}, \dots, \cM \models \bigwedge \Gamma_s\sigma_{\delta}
 $$
 Let $\delta$ be a $\uw$-conditional DAG which is realized in $\cM$  and let it be maximal 
 with this property
  (a $\uw$-conditional DAG $\delta$ is said to be bigger than a $\uw'$-conditional DAG $\delta'$ iff $\uw$ includes  $\uw'$ - the inclusion is as sets, the order is disregarded). Since $\cM\models \phi_{\delta}$ and $\delta$ is realized in $\cM$, it is clear that all $L_{\uw}$-clauses from $S_3$ (hence also  all $L_{\uw}$-literals from $S_1$) are true in $\cM$. Let $\uu=u_1, \dots, u_k$ be the variables from $\ue\setminus \uw$ and let $S_{\uu}$ be the literals from $S_1$ which are not $L_{\uw}$-literals. 
  What we need (by Robinson Diagram Lemma) is to find a model for the following set of literals 
\begin{equation}\label{eq:R}
\Delta_\Sigma(\cM)\cup S_{\uu}\cup \{  w_i= \tilde b_i\mid i=1, \dots, s\} \cup \{  z_i= \tilde a_i\mid i=1, \dots, M\}
\end{equation}
where $\tilde b_i$ is the value of $w_i\sigma_{\delta}$  under the assignment $\uz\mapsto \tilde \ua$. This is the consistency problem we solve in the remaining part of the proof: in order to do so, we obtain by completion a suitable canonical rewriting system that does not introduce inconsistencies. 

We orient the functional equalities in~\eqref{eq:R} from left to right and the equalities $w_i= \tilde b_i$
also from left to right.  The $\ue$ are ordered as $e_1> \cdots > e_N$ and are bigger than the constants naming the elements of $\vert\cM\vert$; function symbols are bigger than constant symbols. 
We show that the ground Knuth-Bendix completion of~\eqref{eq:R} 
 cannot produce any inconsistent literal of the kind $t\neq t$ (this completes the proof of the Theorem).
 
 First notice that the rules $\{  w_i= \tilde b_i\mid i=1, \dots, s\} \cup \{  z_i= \tilde a_i\mid i=1, \dots, M\}$ simply eliminates the $\uw$ and the $\uz$  from $S_{\uu}$ (they become inactive after such normalization steps). Let $\tilde S_{\uu}$ be the set of equalities resulting after this elimination. It turns out that $\tilde S_{\uu}$ can only contain equalities of the kinds
 \begin{equation}\label{eq:11}
 f(a_1, \dots, a_h)= a
\end{equation}
\begin{equation}\label{eq:21}
 u_j\neq a
\end{equation}
where $a_1, \dots a_h, a$ can be either among the $\uu$ or constants naming elements of $\vert \cM\vert$. 
However some of the $\uu$ must be among $a_1, \dots, a_h$ for each equality of the kind~\eqref{eq:11}  because atoms not containing the $\uu$ are removed by $\Delta_\Sigma(\cM)$ and atoms like $u_i=t$ (where $t$ does not contain any of the $\uu$) cannot be there because $\delta$ is maximal. 
During completion, 
 in addition to these kinds of atoms, only atoms of the kind 
 \begin{equation}\label{eq:31}
  u_i= u_j 
 \end{equation}
can possibly be produced. This is a consequence of the next Lemma.
Below we say that a tuple of atoms $\Gamma$ is \emph{realized in $\cM$} iff the $\Gamma$ are $L_{\uw}$-atoms and $\cM\models \bigwedge \Gamma\sigma_{\delta}$; similarly we say that a literal $\Theta$ is \emph{conditionally realized in $\cM$} if there exists  $\Gamma$ realized in $\cM$ with $\Gamma\to \Theta\in S_3$ (if $\Theta$ is a negative literal, $\Gamma\to \Theta$ stands for $\Gamma,\neg \Theta\to\ $).

\begin{lemma}\label{lem:fact}
 Suppose that a literal $\Lambda$\ is produced during the completion of~$\tilde S_{\uu}$. Then 
 it must be of the kinds~\eqref{eq:11},~\eqref{eq:21},~\eqref{eq:31}. Moreover there exists a literal $\Lambda'$ such that (i) $\Lambda'$ is conditionally realized in $\cM$;  (ii) $\Lambda$ is obtained from $\Lambda'$ by rewriting $\uz, \uw$ respectively to $\tilde \ua, \tilde \ub$. 
\end{lemma}

\begin{proof}
By straightforward case analysis; we analyze the most interesting case given by the superposition of two rules of the kind~\eqref{eq:11}. Suppose that
\begin{equation*}
 f(a_1, \dots, a_h)= a~~{\rm and}~~f(a_1, \dots, a_h)= b
\end{equation*}
produce the equality  $a=b$. Then by induction hypothesis, there are in $S_3$ two clauses like
$$
\Gamma'\to f(a'_1, \dots, a'_h)= a', ~~~\Gamma''\to f(a''_1, \dots, a''_h)= a''
$$
with $\Gamma', \Gamma''$ realized in $\cM$, with $a'_1, \dots, a'_h,a'$ rewritable (using 
$\uz, \uw \mapsto \tilde \ua, \tilde \ub$) to $a_1, \dots, a_h,a$, respectively, and with $a''_1, \dots, a''_h,a''$ also rewritable (using 
$\uz, \uw \mapsto \tilde \ua, \tilde \ub$) to $a_1, \dots, a_h,b$, respectively. By Lemma~\ref{lem:pers}, $S_3$ contains a clause subsuming 
\begin{equation}\label{eq:t}
\Gamma', \Gamma'', a_1'=a''_1, \dots, a'_h= a''_h \to a'=a'' 
\end{equation}
which is as required because $\Gamma', \Gamma'', a_1'=a''_1, \dots, a'_h= a''_h$ is realized in $\cM$ and $a'=a''$ rewrites (using 
$\uz, \uw \mapsto \tilde \ua, \tilde \ub$) to $a=b$. It remains to check that $a=b$ is of the kind~\eqref{eq:31}. If both $a',a''$ taken from the consequent of~\eqref{eq:t} belong to $\uz\cup \uw$, then since the antecedent of~\eqref{eq:t} is realized in $\cM$
and~\eqref{eq:t} belongs to $S_3$,
 $a$ and $b$ must be the same element from $\vert \cM\vert$, so that $a=b$ is a trivial identity (which does not enter into the completion). It cannot be that only one between
$a'$ and $a''$ belongs to $\uz\cup \uw$ (the other one being from $\uu$) because $\delta$ is maximal among conditional DAGs realized by $\cM$
and thus it  cannot be properly enlarged by adding to it the additional conditional definition which would be supplied by~\eqref{eq:t}. Thus it must be the case that both $a', a''$ are from $\uu$, which implies that they cannot be rewritten  (using 
$\uz, \uw \mapsto \tilde \ua, \tilde \ub$), so that $a'$ is $a$, $a''$ is $b$ and $a=b$ is of the 
kind~\eqref{eq:31}. 
\end{proof}
\emph{Proof of Theorem~\ref{thm:main} (continued)}. Once $\tilde S_{\uu}$ (standing alone) is completed, only literals of the kinds~\eqref{eq:11},~\eqref{eq:21},~\eqref{eq:31} are produced.
No completion inference is possible between literals  of the kinds~\eqref{eq:11},~\eqref{eq:21},~\eqref{eq:31}
on one side and literals 
from  $\Delta_\Sigma(\cM)\cup  \{  w_i= \tilde b_i\mid i=1, \dots, s\} \cup \{  z_i= \tilde a_i\mid i=1, \dots, M\}$ on the other side; hence the completion of  $\tilde S_{\uu}$ alone, once joined to $\Delta_\Sigma(\cM)\cup  \{  w_i= \tilde b_i\mid i=1, \dots, s\} \cup \{  z_i= \tilde a_i\mid i=1, \dots, M\}$ yields a completion of~\eqref{eq:R}. The only possible inconsistencies that can arise are given by  literals of the kind $u_i\neq u_i$.
Suppose that indeed one such a literal $u_i\neq u_i$ is produced during the completion of $\tilde S_{\uu}$. Applying the above lemma,
there should be 
in $S_3$
a clause like $\Gamma, u_i=u_i\to \ $  (i.e.  after simplification, a clause like $\Gamma \to \ $) with $\Gamma$ being realized in $\cM$. The last 
means that $\cM\models \bigwedge\Gamma\sigma_{\delta}$. This cannot be, because $\Gamma \to \ $ is a $L_{\uw}$-clause from $S_3$: in fact, we have that $\cM\models \phi_{\delta}$ and that $\delta$ is realized in $\cM$, which imply that $\cM\models C\sigma_{\delta}$ holds for every $L_{\uw}$-clause $C$ 
from $S_3$ by the definition of $\phi_{\delta}$. In particular, we should have  $\cM\models \neg\bigwedge\Gamma\sigma_{\delta}$, taking $\Gamma \to \ $ as $C$.
\end{proof}

\section{Conclusions}\label{sec:concl}
Two different algorithms for computing uniform interpolants (UIs) from a formula
in \EUF with a list of symbols to be eliminated are presented. They share a common subpart as well as they
are different in their overall objectives. The first algorithm 
generates a UI expressed as a disjunction of
conjunctions of literals, whereas the second algorithm gives a
compact representation of a UI as a conjunction of Horn clauses.
The output of both algorithms
needs
to be expanded if a fully (or partially) unravelled 
 uniform interpolant is needed for an
application. This restriction/feature is similar in spirit to
syntactic unification where also efficient
unification algorithms never produce output in fully
expanded form to avoid an exponential blow-up. 

For generating a compact representation of the UI, both
algorithms make use of DAG representations of terms by introducing
new symbols to stand for subterms arising in the full expansion of the
UI. Moreover, the second algorithm uses a conditional DAG, a
new data structure introduced in the paper, to represent subterms under conditions.

The complexity of the algorithms is also analyzed. It is shown
that the first algorithm generates exponentially many
branches with each branch of at most quadratic length; 
the UIs produced by the second algorithm have
polynomial size in all the hand-made examples we tried
(but the worst case size is still exponential
as witnessed by \emph{ad hoc} examples like Example~\ref{ex:exp}). A fully expanded UI  can easily be
of exponential size.
%
An implementation of both the algorithms, along with a
comparative study 
%
are planned as future work. In parallel with the
implementation, a characterization of classes of formulae for
which computation of UIs requires polynomial time in our
algorithms (especially in the second one) needs further investigation.

\smallskip
\noindent
\textbf{Acknowledgments.} The third author has been partially supported by the National Science Foundation award CCF -1908804.



\bibliographystyle{alphaurl}
\bibliography{mcmt}

\newcommand{\etalchar}[1]{$^{#1}$}
\begin{thebibliography}{CGG{\etalchar{+}}20b}

\bibitem[BD94]{dill}
J.~R. Burch and D.~L. Dill.
\newblock Automatic verification of pipelined microprocessor control.
\newblock In David~L. Dill, editor, {\em Proc. of {CAV}}, volume 818 of {\em
  LNCS}, pages 68--80. Springer, 1994.
\newblock \href {https://doi.org/10.1007/3-540-58179-0\_44}
  {\path{doi:10.1007/3-540-58179-0\_44}}.

\bibitem[BJ15a]{bonacina2}
M.~P. Bonacina and M.~Johansson.
\newblock Interpolation systems for ground proofs in automated deduction: a
  survey.
\newblock {\em J. Autom. Reasoning}, 54(4):353--390, 2015.

\bibitem[BJ15b]{bonacina1}
M.~P. Bonacina and M.~Johansson.
\newblock On interpolation in automated theorem proving.
\newblock {\em J. Autom. Reasoning}, 54(1):69--97, 2015.

\bibitem[BN98]{BaNi98}
F.~Baader and T.~Nipkow.
\newblock {\em Term Rewriting and All That}.
\newblock Cambridge University Press, United Kingdom, 1998.

\bibitem[CGG{\etalchar{+}}19a]{BPM19}
D.~Calvanese, S.~Ghilardi, A.~Gianola, M.~Montali, and A.~Rivkin.
\newblock Formal modeling and {SMT}-based parameterized verification of
  data-aware {BPMN}.
\newblock In {\em Proc.\ of {BPM}}, volume 11675 of {\em LNCS}, pages 157--175.
  Springer, 2019.

\bibitem[CGG{\etalchar{+}}19b]{CGGMR19}
D.~Calvanese, S.~Ghilardi, A.~Gianola, M.~Montali, and A.~Rivkin.
\newblock From model completeness to verification of data aware processes.
\newblock In {\em Description Logic, Theory Combination, and All That}, volume
  11560 of {\em LNCS}, pages 212--239. Springer, 2019.

\bibitem[CGG{\etalchar{+}}19c]{cade19}
D.~Calvanese, S.~Ghilardi, A.~Gianola, M.~Montali, and A.~Rivkin.
\newblock Model completeness, covers and superposition.
\newblock In {\em Proc.\ of {CADE}}, volume 11716 of {\em LNCS (LNAI)}, pages
  142--160. Springer, 2019.

\bibitem[CGG{\etalchar{+}}20a]{IJCAR20}
D.~Calvanese, S.~Ghilardi, A.~Gianola, M.~Montali, and A.~Rivkin.
\newblock {Combined Covers and Beth Definability}.
\newblock In {\em Proc. of {IJCAR}}, volume 12166 of {\em LNCS (LNAI)}, pages
  181--200. Springer, 2020.

\bibitem[CGG{\etalchar{+}}20b]{MSCS20}
D.~Calvanese, S.~Ghilardi, A.~Gianola, M.~Montali, and A.~Rivkin.
\newblock {SMT}-based verification of data-aware processes: a model-theoretic
  approach.
\newblock {\em Math. Struct. Comput. Sci.}, 30(3):271--313, 2020.

\bibitem[CGG{\etalchar{+}}21]{JAR21}
D.~Calvanese, S.~Ghilardi, A.~Gianola, M.~Montali, and A.~Rivkin.
\newblock Model completeness, uniform interpolants and superposition calculus.
\newblock {\em J. Autom. Reasoning}, 65(7):941--969, 2021.

\bibitem[CK90]{CK}
C.-C. Chang and J.~H. Keisler.
\newblock {\em Model Theory}.
\newblock North-Holland Publishing Co., Amsterdam-London, third edition, 1990.

\bibitem[Cra57]{Craig}
W.~Craig.
\newblock Three uses of the {H}erbrand-{G}entzen theorem in relating model
  theory and proof theory.
\newblock {\em J. Symbolic Logic}, 22:269--285, 1957.
\newblock \href {https://doi.org/10.2307/2963594} {\path{doi:10.2307/2963594}}.

\bibitem[DKPW10]{DSilvaKPW10}
V.~D'Silva, D.~Kroening, M.~Purandare, and G.~Weissenbacher.
\newblock Interpolant strength.
\newblock In {\em Proc. of {VMCAI} 2010}, volume 5944 of {\em LNCS}, pages
  129--145. Springer, 2010.

\bibitem[GGK20]{CILC20}
S.~Ghilardi, A.~Gianola, and D.~Kapur.
\newblock Computing uniform interpolants for {EUF} via (conditional)
  {DAG}-based compact representations.
\newblock In {\em Proc. of {CILC}}, volume 2710 of {\em {CEUR} Workshop
  Proceedings}, pages 67--81. CEUR-WS.org, 2020.

\bibitem[GGMR20]{BPM20}
S.~Ghilardi, A.~Gianola, M.~Montali, and A.~Rivkin.
\newblock Petri nets with parameterised data - modelling and verification.
\newblock In {\em Proc. of {BPM}}, volume 12168 of {\em LNCS}, pages 55--74.
  Springer, 2020.

\bibitem[GGMR21]{BPM21}
S.~Ghilardi, A.~Gianola, M.~Montali, and A.~Rivkin.
\newblock {Delta-BPMN: {A} Concrete Language and Verifier for Data-Aware
  {BPMN}}.
\newblock In {\em Proc. of {BPM}}, volume 12875 of {\em LNCS}, pages 179--196.
  Springer, 2021.

\bibitem[GM08]{GM}
S.~Gulwani and M.~Musuvathi.
\newblock Cover algorithms and their combination.
\newblock In {\em Proc.\ of {ESOP}, Held as Part of {ETAPS}}, pages 193--207,
  2008.

\bibitem[GR10]{mcmt}
S.~Ghilardi and S.~Ranise.
\newblock {MCMT:} {A} model checker modulo theories.
\newblock In {\em Proc.\ of {IJCAR}}, pages 22--29, 2010.

\bibitem[GZ02]{GZ}
S.~Ghilardi and M.~Zawadowski.
\newblock {\em Sheaves, games, and model completions}, volume~14 of {\em Trends
  in Logic---Studia Logica Library}.
\newblock Kluwer Academic Publishers, Dordrecht, 2002.
\newblock \href {https://doi.org/10.1007/978-94-015-9936-8}
  {\path{doi:10.1007/978-94-015-9936-8}}.

\bibitem[HKV12]{HoderKV12}
K.~Hoder, L.~Kov{\'{a}}cs, and A.~Voronkov.
\newblock Playing in the grey area of proofs.
\newblock In {\em Proc. of {POPL} 2012}, pages 259--272. {ACM}, 2012.

\bibitem[Hua95]{HuangInterpolant}
G.~Huang.
\newblock Constructing {Craig} interpolation formulas.
\newblock In {\em Computing and Combinatorics {\em COCOON}}, pages 181--190.
  LNCS, 959, 1995.

\bibitem[Kap97]{KapurRTA}
D.~Kapur.
\newblock Shostak's congruence closure as completion.
\newblock In {\em Proc. of {RTA}}, pages 23--37, 1997.
\newblock \href {https://doi.org/10.1007/3-540-62950-5\_59}
  {\path{doi:10.1007/3-540-62950-5\_59}}.

\bibitem[Kap17]{kapur}
D.~Kapur.
\newblock Nonlinear polynomials, interpolants and invariant generation for
  system analysis.
\newblock In {\em Proc.\ of the 2nd International Workshop on Satisfiability
  Checking and Symbolic Computation co-located with {ISSAC}}, 2017.

\bibitem[Kap19]{kapurJSSC}
D.~Kapur.
\newblock Conditional congruence closure over uninterpreted and interpreted
  symbols.
\newblock {\em J. Systems Science {\&} Complexity}, 32(1):317--355, 2019.

\bibitem[KMZ06]{Zarba}
D.~Kapur, R.~Majumdar, and C.~G. Zarba.
\newblock Interpolation for data structures.
\newblock In {\em Proc. of {SIGSOFT FSE}}, pages 105--116, 2006.
\newblock \href {https://doi.org/10.1145/1181775.1181789}
  {\path{doi:10.1145/1181775.1181789}}.

\bibitem[Lyn59]{lyndon}
R.~C. Lyndon.
\newblock An interpolation theorem in the predicate calculus.
\newblock {\em Pacific J. Math.}, 9(1):129--142, 1959.

\bibitem[McM06]{McM}
K.~L. McMillan.
\newblock Lazy abstraction with interpolants.
\newblock In {\em Proc.\ of {CAV}}, pages 123--136, 2006.

\bibitem[NO79]{NO}
G.~Nelson and D.~C. Oppen.
\newblock Simplification by cooperating decision procedures.
\newblock {\em {ACM} Trans. Program. Lang. Syst.}, 1(2):245--257, 1979.

\bibitem[Pit92]{pitts}
A.~M. Pitts.
\newblock On an interpretation of second order quantification in first order
  intuitionistic propositional logic.
\newblock {\em J. Symb. Log.}, 57(1):33--52, 1992.

\bibitem[Pud97]{Pudlak97}
P.~Pudl{\'{a}}k.
\newblock Lower bounds for resolution and cutting plane proofs and monotone
  computations.
\newblock {\em J. Symb. Log.}, 62(3):981--998, 1997.

\bibitem[Sho84]{Shostak}
R.~E. Shostak.
\newblock Deciding combinations of theories.
\newblock {\em J. {ACM}}, 31(1):1--12, 1984.

\bibitem[Wei12]{Weissenbacher12}
G.~Weissenbacher.
\newblock Interpolant strength revisited.
\newblock In {\em Proc. of {SAT} 2012}, volume 7317 of {\em LNCS}, pages
  312--326. Springer, 2012.

\end{thebibliography}

\end{document}